\RequirePackage{etex}
\documentclass[10pt]{article}
\usepackage[utf8]{inputenc}
\usepackage[a4paper,left = 0.8in,right = 0.8in, top = 1in, bottom = 1in]{geometry}
\usepackage{amsmath}
\usepackage{amsthm}
\usepackage{bm}
\usepackage{cancel}
\usepackage{amssymb}
\usepackage{bbm}
\usepackage[natbibapa]{apacite}
\usepackage{dsfont}
\usepackage{commath}
\usepackage{mathtools}
\usepackage{float}
\usepackage{stmaryrd}
\usepackage{cases}
\usepackage{url}
\usepackage{enumitem,linegoal}
\usepackage{subcaption}
\usepackage{graphicx}
\usepackage{color}
\usepackage{algorithm}
\usepackage{booktabs}
\usepackage{array}
\usepackage{microtype}
\usepackage{shuffle}
\usepackage[]{apacite}

\usepackage[dvipsnames]{xcolor}
\usepackage{hyperref}
\definecolor{Chelsea}{RGB}{3,70,148}
\hypersetup{colorlinks,linkcolor=Chelsea,citecolor=Chelsea,urlcolor = Chelsea}
\usepackage[capitalise]{cleveref}

\theoremstyle{plain}
\newtheorem{theorem}{Theorem}[section]
\newtheorem*{theorem*}{Theorem}
\newtheorem{lemma}[theorem]{Lemma}
\newtheorem{corollary}[theorem]{Corollary}
\newtheorem{proposition}[theorem]{Proposition}

\theoremstyle{definition}
\newtheorem{definition}[theorem]{Definition}
\newtheorem*{definition*}{Definition}

\theoremstyle{remark}
\newtheorem{example}{Example}[section]
\newtheorem{remark}{Remark}[section]

\newcommand{\R}{\mathbb{R}}

\newcommand{\E}{\mathbb{E}}
\newcommand{\Prob}{\mathbb{P}}
\newcommand{\cX}{\mathcal{X}}
\newcommand{\cP}{\mathcal{P}}
\newcommand{\cU}{\mathcal{U}}
\newcommand{\cV}{\mathcal{V}}
\newcommand{\cR}{\mathcal{R}}

\newcommand{\cA}{\mathcal{A}}
\newcommand{\cW}{\mathcal{W}}
\newcommand{\1}{\mathbf{1}}

\newcommand{\argmax}{\operatorname*{arg\,max}}
\newcommand{\argmin}{\operatorname*{arg\,min}}

\usepackage{fancyhdr}
\usepackage{authblk}
\title{\vspace{-10mm}
Path-Space Model Risk via Signature-Induced Optimal Transport
\vspace{-3mm}}
\date{July 21, 2026}
\author{Tomoyuki Ichiba\footnote{Email: {\tt ichiba@pstat.ucsb.edu}. Affiliation: {\tt Department of Statistics and Applied Probability, University of California, Santa Barbara, California, 93106-3110}},  Qijin Shi\footnote{Email: {\tt qijin@ucsb.edu}. Affiliation: {\tt Department of Statistics and Applied Probability, University of California, Santa Barbara, California, 93106-3110}}}

\begin{document}
\maketitle

\begin{abstract}
We propose a signature-induced, optimal transport framework for path-space model risk, in which ambiguity between stochastic path laws is factorized through optimal transport costs on signature coordinates under a common coupling. Via an ambient feature-space relaxation, we derive for affine signature scores and affine half-space events explicit robust expectation and probability bounds that depend only on the baseline model. In both cases, the correction is governed by the same effective budget, which depends only on the affine score and the budget vector of the ambiguity set. Moreover, the same quantity leads to a budget-aware sparse signature surrogate method for more general, possibly implicit, or data-driven, path functionals. We illustrate the resulting methodology through controlled stress tests and benchmark model-misspecification problems in finance and insurance.
\end{abstract}
\vspace{2mm}
	


\tableofcontents

\section{Introduction}
\label{sec:intro}

In finance, many valuation and risk-management problems require evaluating expectations or tail probabilities of path-dependent functionals under a chosen path measure $P$. The object may be an expectation or price, such as $\E_P[F(X)]$, or an event probability such as $P(F(X)\ge b)$, where $F$ depends on the whole trajectory of the underlying process $X$. Canonical examples include exotic option pricing, drawdown criteria, finite-horizon ruin events, and first-passage probabilities. For such quantities, model risk depends intrinsically on the path measure: two models may agree under terminal distributions, low-dimensional marginals, or other familiar summaries, while assigning materially different values to the same claim or risk event. For example, the prices of a knockout option could differ by more than $40\%$ under a local volatility model and a jump diffusion model that are consistent with the same vanilla option surface (cf. \cite{cont2006modeluncertainty}); and a Brownian approximation to a Cram\'er--Lundberg reserve process may significantly underestimate the true tail ruin probabilities because jump-driven and diffusion-driven barrier crossings have different mechanisms (cf. \cite{asmussenalbrecher2010,blanchetmurthy2019}). The issue is therefore not merely whether a baseline law is globally close to an alternative law, but which pathwise discrepancies materially affect the payoff or risk event being assessed.

A standard response to model uncertainty is to fix a tractable baseline model $P_0$, specify an ambiguity set $\mathcal U(P_0)$ of plausible alternative laws, and evaluate the worst-case value in this set. For a payoff or risk functional \(F\), one is usually interested in robust quantities such as
\begin{equation}
\sup_{Q\in \mathcal U(P_0)} \mathbb E_Q[F(X)]
\quad\text{or}\quad
\sup_{Q\in \mathcal U(P_0)} Q(F(X)\ge b).
\end{equation}
The ambiguity set may be defined in several ways, for example, through parametric perturbations; calibration, marginal, or moment constraints; divergence neighborhoods; or constraints based on a prescribed discrepancy between the alternative and baseline laws. Among these approaches, optimal transport provides a natural coupling-based template: an alternative law $Q$ is plausible if it can be coupled with $P_0$ at limited transport cost. In such an approach, the transport cost is usually the central choice in modeling. In path spaces, model uncertainty is in general heterogeneous: drift, volatility, jumps, dependence, crossing behavior, occupation patterns, and tail excursions may affect different financial quantities in different ways. A cost strongly related to an event or a functional can be sharp for that event or that functional itself, but such sharpness typically does not transfer to other payoffs, risk measures, or implicit/learned path functionals. Conversely, a generic path metric may be reusable but too coarse, poorly scaled, or difficult to interpret for model-risk accounting. This motivates the search for a path space ambiguity geometry that is reusable across financial quantities, sensitive to heterogeneous forms of model misspecification, and still structured enough to yield explicit robust formulas and implementable calibration procedures.

Path signatures provide a layer of natural representation for this purpose. The time-augmented signature records ordered iterated integrals of a path. Under standard conditions, it uniquely determines the underlying path (cf. \cite{hamblylyons2010, boedihardjo2016signature}). Moreover, linear functionals of truncated signatures approximate broad classes of continuous path functionals (cf. \cite{chevyrevkormilitzin2016}), while expected signatures play a moment-like role and can characterize path laws under suitable conditions (\cite{chevyrevoberhauser2022}). These facts together suggest the following idea of quantifying path space model risk: {\it using signature coordinates to factorize path-law discrepancies and, at the same time, to represent the financial payoffs and events whose robustness is being assessed}. In this paper, we implement this principle through an optimal-transport construction. We recall basic definitions and important properties of path signatures in Section \ref{subsec:signatures}.

We now briefly describe our approach. Fix a finite set $\cW$ of selected signature words and write
\begin{equation}
  \Phi(x)=(\phi_I(x))_{I\in\cW}
\end{equation}
for the corresponding normalized signature-coordinate map of a path $x$ (cf. Definition \ref{def:time-augmented-signature}). Given coordinate budgets $\delta=(\delta_I)_{I\in\cW}$, the corresponding Signature Optimal Transport (Sig-OT) ambiguity set $\cU_\delta(P_0)$ (cf. Definition \ref{def:sigot-ambiguity}) consists of laws $Q$ for which there exists a single path coupling $\pi\in\Pi(Q,P_0)$ with the tractable baseline model $P_0$ such that
\begin{equation}
  \int |\phi_I(x)-\phi_I(y)|\,\pi(dx,dy)\le \delta_I,
  \qquad I\in\cW.
\end{equation}
We then study model risk problems with respect to such ambiguity sets of path measures. More precisely, for a payoff $F$ and an event $A$, we are interested in
\begin{equation}\label{eq:quantities of interest}
  \mathfrak{R}_\delta(F;P_0):=
  \sup_{Q\in\cU_\delta(P_0)}\E_Q[F(X)],
  \qquad
  \mathfrak{P}_\delta(A;P_0):=
  \sup_{Q\in\cU_\delta(P_0)}Q(A).
\end{equation}

Motivated by the universal approximation property of linear signatures (cf. Proposition \ref{thm:signature-universality}), we explicitly quantify the quantities in \eqref{eq:quantities of interest} with respect to affine signature payoffs, as well as events induced by such payoffs. We show that for affine signature scores
\begin{equation}
  s_{\ell,\beta}(x):=\beta+\sum_{I\in\cW}\ell_I\phi_I(x),
\end{equation}
upper bounds for robust corrections are governed by the effective budget
\begin{equation} \label{eq: kappaelldel}
\kappa(\ell,\delta):=\sum_{I\in\cW}|\ell_I|\delta_I.
\end{equation}
For expectations, Proposition \ref{prop:affine-expectation} provides the upper bound 
\begin{equation}
  \mathfrak{R}_\delta(s_{\ell,\beta};P_0)
  \le
  \E_{P_0}[s_{\ell,\beta}(X)]+\kappa(\ell,\delta),
\end{equation}
while for affine half-space events $A_{\ell,b}:=\{s_\ell\ge b\}$, the same scalar gives a shifted-threshold baseline bound in Theorem \ref{thm:scalar-reduction}:
\begin{equation}
  \mathfrak{P}_\delta(A_{\ell,b};P_0)
  \le
  P_0\big(s_\ell(X)\ge b-u_\kappa\big),
\end{equation}
where, in the non-trivial case, $u_\kappa$ is determined monotonically by the effective budget $\kappa(\ell,\delta)$ in \eqref{eq: kappaelldel} and the baseline distribution of the affine signature  $s_\ell(X) = s_{\ell, 0} (X) $ with $\beta = 0$. The power of the above robust bounds is that, once the score and budgets are fixed, the infinite-dimensional optimization problem over the laws in an infinite-dimensional path space is reduced to a computation that only depends on the baseline model, which is usually chosen to be computationally convenient.


The above bounds also motivate a sparse-surrogate workflow for broader path-dependent functionals. By fitting affine signature scores with a coordinate-weighted lasso penalty proportional to $\sum_I\widehat\delta_I|\ell_I|$, we introduce a preliminary budget-aware learning pipeline such that the above robust bound could be applied to sparse affine signature surrogates for general path-dependent functionals in Section \ref{sec:learning}. In Section \ref{sec:numerics}, we finally illustrate the theory with some examples and numerical experiments for fixed sparse affine linear scores/half-space events, as well as general ones, that arise naturally in finance and insurance.

\paragraph{Contribution.}
The contribution of this paper is to develop a reusable signature-induced optimal-transport geometry for model risk on path laws. We introduce coordinate-wise Sig-OT ambiguity sets under one common path coupling, derive exact path-level dual formulations, and obtain explicit relaxed bounds for affine signature functionals and affine half-space events. The same effective budget also leads to a budget-aware sparse affine-signature learning procedure, which extends the framework to more general, possibly implicit, or data-driven path functionals. We illustrate the resulting methodology through controlled stress tests and benchmark model-misspecification problems in finance and insurance.

\paragraph{Connections to the literature.}
Our work is connected to several streams of literature. For the theory of distributionally robust optimization, we refer to \cite{delageye2010,gohsim2010,rahimianmehrotra2022}, with the Wasserstein and optimal-transport formulations studied in \cite{esfahani2018,blanchetmurthy2019,gaokleywegt2023,kuhnshafieewiesemann2025}. For some of their recent applications in finance, we refer to \cite{bartldrapeautangpi2020} and \cite{carassuswiesel2025}. Some general discussions on model risk were made in \cite{cont2006modeluncertainty,glassermanxu2014}. Other approaches to robustness in finance include parameter-uncertainty \cite{lyons1995uncertain}, nonlinear-expectation \cite{denis2006modeluncertainty,nutzsoner2012superhedging,soner2012secondorderbsde}, and martingale optimal transport \cite{beiglbockhenrylaborderepenkner2013,galichonhenrylaborderetouzi2014}. In particular, path-dependent semimartingale optimal transport and calibration to exotic derivative prices are studied in \cite{guoloeper2021}. For rough paths, path signatures, and their applications in finance, we refer to \cite{frizvictoir2010,FH14,bayerdosreishorvathoberhauser2026} for a comprehensive introduction. In addition, higher-rank expected signatures for representations of adapted processes
and adapted topologies were studied in \cite{bonnierliuoberhauser2020}; sparse signature learning via Lasso was studied in \cite{guowangzhangzhao2025lasso}; universal approximation with signatures was developed in \cite{CuchieroSchmockerTeichmann2026,cuchieroprimaverasvalutoferro2025}. In particular, Sig-Wasserstein GANs, developed in \cite{ni2021sigwgan,liao2024sigcwgan}, provide another interesting application of combining optimal transport and path signature in time series generation. Finally, a connection between signature asymptotics and Wasserstein distances between empirical measures was studied in \cite{cassmessadeneturner2023}.

\paragraph{Organization.}
Section~\ref{sec:prelim} introduces the notation, path signatures, and the OT-DRO background. Section~\ref{sec:theory} develops the coordinate-wise Sig-OT geometry and derives the robust affine formulas. Section~\ref{sec:learning} presents the effective-budget learning method. Section~\ref{sec:numerics} examines the approach through controlled affine-signature stress tests and model-misspecification examples in finance and insurance. Section~\ref{sec:conclusion} concludes. Appendix~\ref{app:proofs} proves the multi-budget OT-DRO duality theorem, Appendix~\ref{app:strictness} explains the strictness of the feature-space relaxation, Appendix~\ref{app:analytic-bounds} gives preliminary model-based Sig-OT budget bounds, Appendix~\ref{app:j1-ruin-benchmark} details the $J_1$ benchmark used in Section \ref{subsubsec:reserve-risk}, and Appendix~\ref{app:ablations} collects supplementary budget-geometry diagnostics and ablations.

\section{Preliminaries}
\label{sec:prelim}

Throughout, the horizon $T>0$ is fixed and $\cX$ denotes the (rough) path space under consideration. We write $\cP(\cX)$ for Borel probability laws on $\cX$, $\Pi(P,Q)$ for couplings of $P,Q\in\cP(\cX)$, and $P_0$ for a baseline law on $\cX$. Path-dependent payoffs are measurable maps $F:\cX\to\R$, and path-dependent events are denoted by $A\subseteq\cX$, often in threshold form $A=\{F\ge b\}$.

\subsection{Path signatures}
\label{subsec:signatures}

This subsection fixes the notation for path signatures used throughout the paper, as well as recalls some key properties, following the standard tensor-algebra convention used in the signature and rough-path literature; see, for example, \cite{chevyrevkormilitzin2016} or \cite{bayerdosreishorvathoberhauser2026}.

Let us take $V:=\R^{d+1}$ with canonical basis $e_0,e_1,\ldots,e_d$, where the index $0$ is reserved for time. The extended tensor algebra $T((V))$ and its step-$N$ truncation $T^N(V)$ are defined by  $T((V)):=\prod_{k\ge0}V^{\otimes k}$, $T^N(V):=\bigoplus_{k=0}^N V^{\otimes k}$ with $ V^{\otimes0}:=\R$. The tensor product on $T((V))$ is defined by the concatenation of tensor levels: if $a=(a^{(k)})_{k\ge0}$ and $b=(b^{(k)})_{k\ge0}$, then $(a\otimes b)^{(n)} = \sum_{k=0}^n a^{(k)}\otimes b^{(n-k)}$. Let $\cA=\{0,1,\ldots,d\}$ be the alphabet of coordinate labels and let
$\cA^*=\bigcup_{k\ge0}\cA^k$ be the set of finite words, with the empty word $\emptyset$. For $I = (i_1,\ldots,i_k)\in\cA^k$, write $|I|=k$ and
$e_I=e_{i_1}\otimes\cdots\otimes e_{i_k}$.
Moreover, we write $I^-:=(i_1,\ldots,i_{k-1})$. We identify a tensor $a\in T((V))$ with its word coordinates through
$ a=\sum_{I\in\cA^*} a_I e_I$, $a_I:=\langle e_I,a\rangle$ with the convention that $e_\emptyset$ spans $V^{\otimes0}$. Throughout, we consider the time-augmented signature of a path, to remove tree-like equivalence for the uniqueness of the representation, cf. \cite{hamblylyons2010}.

\begin{definition}[Time-augmented signature]
\label{def:time-augmented-signature}
For each $\mathbb R^d$-valued path $x:[0,T]\to\R^d$ we associate the time parameter and write $\widehat x_t=(t,x_t)\in V = \mathbb R^{d+1}$. If $\widehat x$ is of bounded variation, its signature for the subinterval $[s,t]\subseteq[0,T]$ is the tensor series
\begin{equation}
  S_{s,t}(\widehat x)
  =
  1+
  \sum_{k\ge1}
  \int_{s<u_1<\cdots<u_k<t}
  d\widehat x_{u_1}\otimes\cdots\otimes d\widehat x_{u_k}
  \in T((V)).
\end{equation}
The coordinate associated with a word $I=(i_1,\ldots,i_k) \in \mathcal A^k$ and the level-$N$ truncated signature are defined by 
\[
S^\emptyset_{s,t}(\widehat x)=1, \quad 
  S^I_{s,t}(\widehat x)
  :=
  \langle e_I,S_{s,t}(\widehat x)\rangle
  =
  \int_{s<u_1<\cdots<u_k<t}
  d\widehat x^{i_1}_{u_1}\cdots d\widehat x^{i_k}_{u_k}, \quad 
  S^N_{s,t}(\widehat x)
  :=
  \sum_{|I|\le N}S^I_{s,t}(\widehat x)e_I
  \in T^N(V)
\]
respectively for $0 \le s, t \le T$. When $s=0$ and $t=T$, we write $S(\widehat x)$ and $S^N(\widehat x)$ without subscripts, unless an ambiguity can arise. For  $\R^d$-valued continuous  semimartingale $X$, the time-augmented signature $S(\widehat X)$ is defined analogously by iterated Stratonovich integrals.
\end{definition}

The key properties of path signatures are that they summarize path information in an ordered and algebraically tractable way. Before stating them, we begin by introducing a basic operation in the rough path and signature literature; see also Definition 2.2 in \cite{cuchierogazzanimollersvalutoferro2024}.

\begin{definition}[Shuffle product]
\label{def:shuffle-product}
For two nonempty words $I=(i_1,\ldots,i_k)\in \mathcal A^k$ and $J=(j_1,\ldots,j_m) \in \mathcal A^m$, the shuffle product $e_I\shuffle e_J$ is defined recursively as 
$e_I\shuffle e_J := 
  (e_{I^-}\shuffle e_J)\otimes e_{i_k}
  +
  (e_I\shuffle e_{J^-})\otimes e_{j_m}$ for $k, m \ge 1$, and $e_\emptyset\shuffle e_I:=e_I\shuffle e_\emptyset:=e_I$. The operation extends bilinearly to $a,b \in T((V))$ as 
\begin{equation}
  a\shuffle b
  :=
  \sum_{I,J\in\cA^*}a_Ib_J(e_I\shuffle e_J).
\end{equation}
\end{definition}

We record the standard identities and representation properties needed in the following. These results have been well developed in the rough path literature. We refer to, e.g., \cite{cuchierogazzanisvalutoferro2023}, \cite{FH14}, \cite{chevyrevkormilitzin2016}, and \cite{bayerdosreishorvathoberhauser2026} for more detailed expositions.

\begin{proposition}[Standard signature identities]
\label{prop:signature-identities}
Let $\widehat x$ be a bounded-variation path (or a sample path of a continuous semimartingale, respectively). The following identities hold (or hold almost surely, respectively).
\begin{enumerate}[label=(\roman*)]
  \item Chen's identity:
$S_{s,t}(\widehat x)=S_{s,u}(\widehat x)\otimes S_{u,t}(\widehat x)$,
  $0\le s\le u\le t\le T$.
  \item The word coordinates satisfy the shuffle identity: If $\shuffle$ denotes the shuffle product on words, then
\begin{equation}
  \langle e_I,S_{s,t}(\widehat x)\rangle
  \langle e_J,S_{s,t}(\widehat x)\rangle
  =
  \langle e_I\shuffle e_J,S_{s,t}(\widehat x)\rangle.
\end{equation}
\end{enumerate}
Equivalently, truncated signatures take values in the step-$N$ group-like subset of $T^N(V)$.
\end{proposition}

\begin{proposition}[Uniqueness of time-augmented signatures]
Let $\widehat x_t:=(t, x_t)$ and $\widehat y_t:=(t, y_t)$ be the time-augmentation of a bounded variation path $x$ and $y$, such that $x_0=y_0$. Then the signature $S(\widehat x) = S(\widehat y)$ if and only if $x=y$. Similarly, let $X$ and $Y$ be two continuous semimartingales with $X_0=Y_0=0$. Then $S(\widehat X) = S(\widehat Y)$ if and only if $X=Y$ as a process. 
\end{proposition}
The statements for time-augmented signatures in these propositions are standard; see, for example, \cite{hamblylyons2010,cuchierogazzanimollersvalutoferro2024, boedihardjo2016signature}.

\begin{proposition}[Universal approximation theorem]
\label{thm:signature-universality}
Let $K$ be a compact family of time-augmented paths on a finite closed time interval with respect to the supremum norm. Then, for every continuous functional $F:K\to\R$ and every $\varepsilon>0$, there are a truncation level $N$ and a linear functional $\ell\in T^N(V)$ such that
\begin{equation}
  \sup_{\widehat x\in K}
  \left|F(\widehat x)-\langle \ell,S^N(\widehat x)\rangle\right|
  <\varepsilon.
\end{equation}
A similar statement holds for semimartingale paths in the corresponding rough path space (see, for example, Proposition 2.3 of \cite{cuchierogazzanimollersvalutoferro2024}).
\end{proposition}

The preceding propositions essentially reveal that time-augmented signatures are not merely a collection of ad hoc features but a universally informative ordered feature set. To make this representation more concrete, we spell out the meaning of a few low-order coordinates in the path signature.

\begin{example}\label{ex:meaning of coordinates}
Consider a one-dimensional process $X$ and write $X_{0,t}=X_t-X_0$. The first nontrivial coordinate $S^X=X_{0,T}$
is the terminal shift, which does not depend on the intermediate path information. At level two,
$S^{Xt}=\int_0^T X_{0,t}\, \mathrm dt$, $S^{tX}=T\cdot X_{0,T}-\int_0^T X_{0,t}\,\mathrm dt$, where $S^{Xt}$ records the path-average information, while $S^{tX}$ combines the terminal shift with the same cumulative quantity. Their antisymmetric difference 
$A^{tX}:=S^{tX}-S^{Xt}$ 
is the signed time-state L\'evy area. It separates paths with similar endpoints but different timing: early drawdowns followed by recovery, late losses, or paths that spend much of the horizon below their own time-average produce different signs and magnitudes. Moreover, the pure coordinate $S^{XX}$ equals $\tfrac12 X_{0,T}^2$, so it is a second-order terminal-magnitude feature; higher pure words similarly encode powers of the terminal shift. Mixed higher-order words such as $S^{Xtt}$, $S^{XXt}$, and $S^{XXXt}$ add timing-weighted magnitude and further asymmetry information.

\end{example}

\subsection{Distributional model risk via optimal transport}
\label{subsec:otdro}

Now we recall the second main ingredient: distributional model risk based on optimal transport (see, for example, \cite{blanchetmurthy2019} and \cite{gaokleywegt2023}). Let $P_0$ be a tractable baseline probability measure on a Polish space $E$, and define the ambiguity set $\mathcal{U}(P_0)$ as a closed neighborhood of $P_0$ with respect to some metric, consisting of other plausible models that one is not willing to rule out. For a fixed payoff function $F$, the robust value $\sup_{Q\in \mathcal{U}(P_0)}\mathbb{E}_Q[F]$ therefore measures the sensitivity of an expectation to misspecification of the underlying law. It is often difficult to compute the primal problem $\sup_{Q\in \mathcal{U}(P_0)}\mathbb{E}_Q[F]$ directly as an infinite-dimensional optimization. A main appeal of optimal-transport ambiguity sets is that they admit dual reformulations, which are often easier to analyze and compute. We briefly recall this approach, following \cite{blanchetmurthy2019}, but in a slightly more general setting. 

Let $E$ be a Polish space, let $P_0\in\cP(E)$, and let $c_1,\ldots,c_m:E\times E\to[0,\infty]$ be lower semicontinuous transport cost functions satisfying $c_j(x,x)=0$. Write $\mathsf C:=\left\{(z,x)\in E\times E:c_j(z,x)<\infty,\ j=1,\ldots,m\right\}$
for the common finite-cost relation. For $\lambda\in\R_+^m$, define the aggregate cost $c_\lambda(z,x) := \sum_{j=1}^m\lambda_jc_j(z,x)$, if $(z,x)\in\mathsf C$, and $+\infty$ otherwise. Here, pairs with an infinite cost in any coordinate remain inadmissible even when the corresponding multiplier is zero. For $\delta\in\R_+^m$, define the ambiguity set 
\begin{equation}
  \cU_\delta^{\mathbf c}(P_0)
  :=
  \left\{
    Q\in\cP(E):
    \text{ there exists } \pi\in\Pi(Q,P_0)
    \text{ such that }
    \int c_j(z,x)\,\pi(dz,dx)\le\delta_j,\ j=1,\ldots,m
  \right\}.
\end{equation}

\begin{theorem}[Multi-budget OT-DRO duality]
\label{thm:finite-budget-otdro}
If $f:E\to\R$ is upper semicontinuous and $f\in L^1(P_0)$, then
\begin{equation}
  \sup_{Q\in\cU_\delta^{\mathbf c}(P_0)}
  \int f\,dQ
  =
  \inf_{\lambda\in\R_+^m}
  \left\{
    \sum_{j=1}^m\lambda_j\delta_j
    +
    \int
    \sup_{z\in E}
    \left(
      f(z)-c_\lambda(z,x)
    \right)P_0(dx)
  \right\} 
\end{equation}
for $\delta\in(0,\infty)^m$ with equality understood in the extended real sense. The inner supremum is universally measurable. If the common value is finite, then the dual infimum is attained by some $\lambda^\star\in\R_+^m$.
\end{theorem}

\begin{proof}
    See Appendix~\ref{app:proofs}.
\end{proof}

\begin{corollary}[Multi-budget worst-case probability]
\label{cor:finite-budget-probability}
    Let $\delta\in(0,\infty)^m$ and let $A\subseteq E$ be a closed event. Then
    \begin{equation}
  \sup_{Q\in\cU_\delta^{\mathbf c}(P_0)}Q(A)
  =
  \inf_{\lambda\in\R_+^m}
  \left\{
    \sum_{j=1}^m\lambda_j\delta_j
    +
    \int (1-d_\lambda(x,A))_+\,P_0(dx)
  \right\},
\end{equation}
where
  $d_\lambda(x,A)
  :=
  \inf_{z\in A}c_\lambda(z,x)$ is the infimum of the aggregate cost $c_\lambda$ of $x \in E$ against the event $A$.

\end{corollary}
\begin{proof}
    Apply Theorem~\ref{thm:finite-budget-otdro} with $f=\1_A$ and $\sup_{z \in E} (\1_A (z) - c_\lambda (z, x) ) = ( 1 - d_\lambda (x, A))_+ $, $x \in E$. 
\end{proof} 

The significance of these dual formulations is that they depend only on the baseline measure $P_0$, which is completely characterized and is usually chosen in a way that is easy to work with or draw samples from, reducing the primal infinite-dimensional optimization problem to finite-dimensional optimization problems.

\section{Signature-Induced OT Model Risk}
\label{sec:theory}

The above multi-budget OT-DRO theorem (Theorem \ref{thm:finite-budget-otdro}) is applicable for a large class of transport cost functions $c_1, \ldots, c_m$. The critical modeling specification is the design of the transport costs, which inherently determine the geometry of the ambiguity set: the costs determine which changes of a stochastic law are considered small and hence which forms of model misspecification the robust value protects against. For an explicitly given family of path functionals of interest $F_1,\ldots,F_m$, one natural choice is to tailor the transport costs as
\begin{equation}
  c_j(x,y):=|F_j(x)-F_j(y)|,\qquad x, y \in E, j=1,\ldots,m.
\end{equation}
The resulting ambiguity set $\mathcal U_\delta^{\mathbf c}(P_0)$ then directly controls the perturbations of these particular functionals. 

This can be effective when the relevant quantities are known in advance; however, general path-dependent quantities of interest may be implicit, learned, or data-driven, and thus, such a payoff-specific ambiguity set can be difficult to interpret and calibrate. Moreover, such a construction does not provide a reusable geometry for other path-dependent payoffs or probabilities. As discussed in Section~\ref{subsec:signatures}, path signatures offer a natural structure to resolve these issues: Linear functionals of signatures approximate broad classes of path functionals, while individual signature coordinates encode ordered information about the path. This suggests using selected signature coordinates not only to represent payoff but also to measure discrepancies between path laws. In this section, we develop the signature-induced optimal-transport model-risk framework. We first formalize the coordinate-wise Sig-OT ambiguity geometry in Section~\ref{subsec:geometry}, and then analyze robust expectations and worst-case probabilities for linear signature scores in Sections~\ref{subsec:robust-expectation}-\ref{subsec:robust-probability}.

\subsection{Coordinate-wise Sig-OT ambiguity geometry}
\label{subsec:geometry}
We now instantiate the multi-budget OT-DRO by selected signature coordinates as transport coordinates. Fix the truncation level $N$ and a finite selected word set $\cW\subset\cA^{\le N}\setminus\{\emptyset\}$. Since different signature coordinates may have different scales, we choose (usually model/sample-determined) deterministic normalizers $r_I>0$ and define
$\phi_I(x):={S^I_{0,T}(\widehat x)}/{r_I}$, $x \in \mathcal X $, $I\in\cW$ and the selected signature feature map $\Phi(x):=(\phi_I(x))_{I\in\cW}\in\R^{\cW}$, $x \in \mathcal X$.

\begin{definition}[Coordinate-wise Sig-OT ambiguity set]
\label{def:sigot-ambiguity}
For a budget vector $\delta=(\delta_I)_{I\in\cW}\in\R_+^{\cW}$, define
\begin{equation}
  \cU_\delta(P_0)
  :=
  \Big\{
    Q\in\cP(\cX):
    \text{there exists }  \pi\in\Pi(Q,P_0)
    \text{ such that }
    \int c_I(x,y)\,\pi(dx,dy)\le\delta_I,
    \ I\in\cW
  \Big\}, 
\end{equation}
where we define the coordinate transport cost
$c_I(x,y):=|\phi_I(x)-\phi_I(y)|$, $x,y\in\cX$, $I\in\cW$ from $\Phi$. 
\end{definition}

\begin{remark}
We note that although the constraints are coordinate-wise, a law $Q$ is feasible only if a single coupling $\pi\in\Pi(Q,P_0)$ satisfies all coordinate budgets simultaneously. This common-coupling requirement underlies the multi-budget duality below. It is distinct from multi-source Wasserstein DRO in \cite{rychener2024heterogeneous}, which constrains a candidate law to lie in the intersection of Wasserstein neighborhoods around several empirical source laws, with each neighborhood certified through its corresponding source coupling. Here, all budgets are instead imposed on selected signature coordinates relative to one baseline law, and are jointly certified by the same pathwise transport plan.
\end{remark}

For the tractable formulas that we will develop below, we shall consider quantities that depend on a path through the selected signature feature vector $\Phi(x)$. We thus pass to the feature laws. This passage introduces an important distinction. Recall from Proposition \ref{prop:signature-identities} that different coordinates of the path signature satisfy the shuffle identity, meaning that the vector $\Phi(x)$ in general lies only in a constrained subset of $\R^{\cW}$. Thus, not every law on $\R^{\cW}$ is generated by actual paths. It is important to separate the realizable feature laws from the larger ambient feature-space relaxation.

\begin{definition}[Realizable and relaxed feature-law ambiguity sets]
\label{def:feature-law-sets}
Let $\mu_0:=\Phi_\#P_0\in\cP(\R^{\cW})$ be the pushforward of $P_0$ by $\Phi$.
We define the set of realizable selected-signature feature laws as $\cR_\Phi:=\{\Phi_\#Q:Q\in\cP(\cX)\}\subseteq\cP(\R^{\cW})$. The relaxed common-coupling feature-space transport set is then defined as
\begin{equation}
  \cV_\delta(\mu_0)
  =
  \Big\{
    \nu\in\cP(\R^{\cW}):
    \text{there exists } \gamma\in\Pi(\nu,\mu_0)
    \text{ such that }
    \int |z_I-z'_I|\,\gamma(dz,dz')\le\delta_I,
    \ I\in\cW
  \Big\}
\end{equation}
for budget vector $\delta=(\delta_I)_{I \in \mathcal W} \in \mathbb R_+^{\mathcal W}$. 
\end{definition}

The set $\cR_\Phi$ records which selected-feature laws can be realized by some path law. The set $\cV_\delta(\mu_0)$ forgets the realizability of the paths  and keeps only the coordinate-wise transport budgets in feature space. These two views are linked by the feature-law identity
\begin{equation}
  \label{eq:feature-law-identity}
  \{\Phi_\#Q:Q\in\cU_\delta(P_0)\}
  =
  \cR_\Phi\cap\cV_\delta(\mu_0)
\end{equation}
The identity follows by transferring couplings through $\Phi$. A feasible path coupling $\pi\in\Pi(Q,P_0)$ pushes forward under $(\Phi,\Phi)$ to a feasible feature coupling between $\Phi_\#Q$ and $\mu_0$. Conversely, if $\nu=\Phi_\#Q$ and $\gamma\in\Pi(\nu,\mu_0)$ satisfies the feature-space budget constraints, then lifting $\gamma$ through regular conditional laws of $Q$ and $P_0$ given $\Phi$ gives a path coupling that preserves the same coordinate-budget values.

Thus, for functionals depending only on $\Phi(x)$, optimizing over $\cR_\Phi\cap\cV_\delta(\mu_0)$ is equivalent to the original path-level problem. The removal of the realizability requirement and the optimization over the whole $\cV_\delta(\mu_0)$ gives a tractable relaxation of the ambient feature-space. Because truncated signatures satisfy algebraic relations, this relaxation can be strict; Appendix~\ref{app:strictness} exhibits an explicit coordinate-wise step-2 gap.

\paragraph{On the coordinate budgets.}
The budget vector $\delta$ is treated as an input to the Sig-OT ambiguity set and can be empirically calibrated from the baseline and stress samples. On the theoretical side, the finite first moments of the selected signature coordinates of the baseline and the stress model always allow finite coordinate budgets $\delta$ under a chosen common coupling to cover the stress model. Moreover, we record two preliminary model-based bounds in Appendix~\ref{app:analytic-bounds}: synchronously coupled diffusions lead to budgets controlled by initial-condition and coefficient perturbations, while fixed-grid Gaussian path models lead to budgets controlled by the corresponding finite-dimensional Gaussian transport distance after interpolation.

\subsection{Robust expectations for affine signature scores}
\label{subsec:robust-expectation}

We consider the path-level robust expectation 
  $\mathfrak{R}_\delta(F;P_0)
  := \sup_{Q\in\cU_\delta(P_0)}\E_Q[F(X)]$ 
for a measurable payoff $F:\cX\to\R$. For $\delta\in(0,\infty)^{\cW}$ and the upper semicontinuous integrable $F$, by taking
  $c_I(x,y)=|\phi_I(x)-\phi_I(y)|$, $I\in\cW$,
in Theorem~\ref{thm:finite-budget-otdro}, the robust expectation admits the dual formulation:
\begin{equation}
  \mathfrak{R}_\delta(F;P_0)
  =
  \inf_{\lambda\in\R_+^{\cW}}
  \Big\{
    \sum_{I\in\cW}\lambda_I\delta_I
    +
    \E_{P_0}\Big[
      \sup_{y\in\cX}
      \Big(
        F(y)-\sum_{I\in\cW}\lambda_I|\phi_I(y)-\phi_I(X)|
      \Big)
    \Big]
  \Big\}.
\end{equation}

This dual formulation is still hard to analyse for general functionals $F$. Recall from Proposition~\ref{thm:signature-universality} that, on compact path classes, continuous path functionals can be approximated by linear functionals of truncated signatures. This motivates studying the dual formulation first for affine signature scores. For $\ell=(\ell_I)_{I\in\cW}$ and $\beta\in\R$, define
$s_{\ell,\beta}(x):=\beta+\sum_{I\in\cW}\ell_I\phi_I(x)$, and set  $\kappa(\ell,\delta):=\sum_{I\in\cW}|\ell_I|\delta_I$.

\begin{proposition}[Affine robust-expectation premium]
\label{prop:affine-expectation}
Assume $\mu_0=\Phi_\#P_0$ has finite first moments. Then
\begin{equation} \label{eq:prop3.3}
\begin{aligned}
  \mathfrak{R}_\delta(s_{\ell,\beta};P_0)
  & =
  \sup_{\nu\in\cR_\Phi\cap\cV_\delta(\mu_0)}
  \int \Big(\beta+\sum_{I\in\cW}\ell_I z_I\Big)\nu(dz) \\
  & \le
  \sup_{\nu\in\cV_\delta(\mu_0)}
  \int \Big(\beta+\sum_{I\in\cW}\ell_I z_I\Big)\nu(dz) =
  \E_{P_0}[s_{\ell,\beta}(X)]+\kappa(\ell,\delta).
\end{aligned}
\end{equation}
\end{proposition}

\begin{proof}
The first equality follows from the feature-law identity \eqref{eq:feature-law-identity}, because $s_{\ell,\beta}(x)=\beta+\sum_I\ell_I\Phi(x)_I$. Dropping the realizability condition gives the relaxed upper bound. For any $\nu\in\cV_\delta(\mu_0)$ and  $\gamma\in\Pi(\nu,\mu_0)$,
\begin{equation}
\begin{aligned}
  \int \Big(\beta+\sum_I\ell_Iz_I\Big)\nu(dz)
  -
  \int \Big(\beta+\sum_I\ell_Iz_I'\Big)\mu_0(dz')
  &=
  \int\sum_{I\in\cW}\ell_I(z_I-z_I')\,\gamma(dz,dz') \\
  &\le
  \sum_{I\in\cW}|\ell_I|
  \int |z_I-z_I'|\,\gamma(dz,dz') \le
  \kappa(\ell,\delta).
\end{aligned}
\end{equation}
This proves the relaxed upper bound. The bound is attained in the ambient feature space by the translation $T(z)_I=z_I+\delta_I\operatorname{sgn}(\ell_I)$, $I\in\cW$, with $\operatorname{sgn}(0)=0$: if $\nu^\star=T_\#\mu_0$, then the coupling $(T(z),z)_\#\mu_0$ certifies $\nu^\star\in\cV_\delta(\mu_0)$ and the objective increases by exactly $\sum_I|\ell_I|\delta_I$.
\end{proof}

In other words, $\kappa(\ell,\delta)$ is the exact robust premium for the ambient relaxed feature-space problem, and a relaxed premium for the original path-level problem. Besides, for a fixed budget vector, this premium scales with the weighted support of the coefficient vector: sparse affine scores have smaller premium.

\subsection{Worst-case probability bounds for signature half-space events}
\label{subsec:robust-probability}

We next analyze worst-case probabilities. For a closed event $A\subseteq\cX$, define
  $\mathfrak{P}_\delta(A;P_0)
:=\sup_{Q\in\cU_\delta(P_0)}Q(A)$. For $\delta\in(0,\infty)^{\cW}$, the exact path-level dual follows from Corollary~\ref{cor:finite-budget-probability} with the same coordinate costs $c_I(x,y)=|\phi_I(x)-\phi_I(y)|$:
  $\mathfrak{P}_\delta(A;P_0)
  =
  \inf_{\lambda\in\R_+^{\cW}}
  \{  \sum_{I\in\cW}\lambda_I\delta_I
    +
    \E_{P_0}\left[(1-d_\lambda(X,A))_+\right]
  \}$,
where
$d_\lambda(x,A)
  =
  \inf_{y\in A}
  \sum_{I\in\cW}\lambda_I|\phi_I(y)-\phi_I(x)|$.

As in the expectation case, this dual formula is exact at the path level but is not usually explicit. In many examples, such as ruin or occupation time type events, the risk event $A$ in consideration can be abstracted into a functional half-space event of the form $A=\{F\geq b\}$, where $F$ is a suitable path functional. Again, motivated by the universal approximating property of affine signatures (cf. Proposition \ref{thm:signature-universality}), we focus here on affine signature half-spaces, for which the relaxed feature-space bound reduces to a one-dimensional optimization and, under mild regularity, to a shifted baseline threshold. For $\ell=(\ell_I)_{I\in\cW}$ and $b\in\R$, define  $s_\ell(x):=\sum_{I\in\cW}\ell_I\phi_I(x)$,   $A_{\ell,b}:=\{x\in\cX:s_\ell(x)\ge b\}$,
and the corresponding ambient feature half-space
  $H_{\ell,b}
  :=
  \{
    z\in\R^{\cW}:
    \sum_{I\in\cW}\ell_Iz_I\ge b\}$.
Since $A_{\ell,b}=\Phi^{-1}(H_{\ell,b})$, the feature-law identity \eqref{eq:feature-law-identity} gives   $\mathfrak{P}_\delta(A_{\ell,b};P_0)
  =
\sup_{\nu\in\cR_\Phi\cap\cV_\delta(\mu_0)}\nu(H_{\ell,b})
  \le
\sup_{\nu\in\cV_\delta(\mu_0)}\nu(H_{\ell,b})$. The only geometric input needed is the weighted distance from a point to a half-space. For $\lambda\in(0,\infty)^{\cW}$, define
  $\|h\|_\lambda:=\sum_{I\in\cW}\lambda_I|h_I|$, $
  \|\ell\|_{\lambda,*}:=\max_{I\in\cW}{|\ell_I|}\,/\,{\lambda_I}$.
\begin{lemma}[Weighted distance to an affine halfspace]
\label{lem:weighted-halfspace-distance}
For $\ell\ne0$, $\lambda\in(0,\infty)^{\cW}$, and $z\in\R^{\cW}$,
\begin{equation}
  \inf_{w\in H_{\ell,b}}\|w-z\|_\lambda
  =
  \frac{1}{\|\ell\|_{\lambda,*}} \big[b-\sum_{I\in\cW}\ell_Iz_I\big]_+.
\end{equation}
\end{lemma}

\begin{proof}
If $z\in H_{\ell,b}$, both sides are zero;  Otherwise, set
  $r=b-\sum_{I\in\cW}\ell_Iz_I>0$.
For any $w\in H_{\ell,b}$,
  $r
  \le
  \sum_I\ell_I(w_I-z_I)
  \le
  \|\ell\|_{\lambda,*}\|w-z\|_\lambda$,
which gives the lower bound. 

For the reverse inequality, choose a maximize 
  $J\in\argmax_{I\in\cW}({|\ell_I|}/{\lambda_I})$ for which $\ell_J\ne0$.   
Define $w$ by $w_I=z_I$ for $I\ne J$ and
  $w_J=z_J+\operatorname{sgn}(\ell_J)\cdot {r}/{|\ell_J|}$. Then $\sum_I\ell_Iw_I=b$, so $w\in H_{\ell,b}$, and
  $\|w-z\|_\lambda
  =
  \lambda_J{r}/{|\ell_J|}
  =
  {r}/{\|\ell\|_{\lambda,*}}$. This proves equality. Note $w_J$ is generally no longer a valid signature vector. 
\end{proof}

Recall that $\kappa(\ell,\delta)=\sum_{I\in\cW}|\ell_I|\delta_I$. Define
  $D:=[b-s_\ell(X)]_+$,
  $X\sim P_0$,
and
  $h(u):=\E_{P_0}[D\,\1_{\{D\le u\}}]$,$
  u\ge0$.

\begin{theorem}[Affine halfspace probability bound]
\label{thm:scalar-reduction}
Assume $\delta\in(0,\infty)^{\cW}$, $\ell\ne0$, and $\E_{P_0}[D]<\infty$. Then
\begin{equation} \label{eq:thm3.5}
  \mathfrak{P}_\delta(A_{\ell,b};P_0)
  \le
  \sup_{\nu\in\cV_\delta(\mu_0)}\nu(H_{\ell,b}) =
  \inf_{u>0}
  \left\{
    \frac{\kappa(\ell,\delta)}{u}
    +
    \E_{P_0}\left[\left(1-\frac{D}{u}\right)_+\right]
  \right\}.
\end{equation}
If $\kappa(\ell,\delta)\ge\E[D]$, the relaxed value is $1$; In the interior case $0<\kappa(\ell,\delta)<\E[D]$, if $u^\star>0$ satisfies $h(u^\star)=\kappa(\ell,\delta)$
and $P_0(D=u^\star)=0$, then $u^\star$ minimizes the scalar objective above and
  \begin{equation}
      \sup_{\nu\in\cV_\delta(\mu_0)}\nu(H_{\ell,b})
  =
  P_0\big(s_\ell(X)\ge b-u^\star\big).
  \end{equation}
\end{theorem}

\begin{proof}
Apply Corollary \ref{cor:finite-budget-probability} in the ambient feature space $E=\R^{\cW}$ with baseline law $\mu_0$, event $H_{\ell,b}$, and coordinate costs $|z_I-z_I'|$. The half-space-distance formula in Lemma \ref{lem:weighted-halfspace-distance} gives
\begin{equation} \label{eq:thm3.5prf}
  \sup_{\nu\in\cV_\delta(\mu_0)}\nu(H_{\ell,b})
  =
  \inf_{\lambda\in(0,\infty)^{\cW}}
  \Big\{
    \sum_{I\in\cW}\lambda_I\delta_I
    +
    \E_{P_0}\Big[
      \left(
        1-\frac{D}{u(\lambda)}
      \right)_+
    \Big]
  \Big\},
\end{equation}
where $u(\lambda)=\max_I |\ell_I|/\lambda_I$. For fixed $u_0>0$, the constraint $u(\lambda)=u_0$ implies
$\lambda_I\ge |\ell_I|/u_0$ for $I\in\cW$ and, hence, with 
$\delta_I\ge0$, we have 
  $\sum_{I\in\cW}\lambda_I\delta_I
  \ge
  {u_0}^{-1} \sum_{I\in\cW}|\ell_I|\delta_I
  =
  {\kappa(\ell,\delta)}/{u_0}$.
This lower bound is attained in coordinates with $\ell_I\ne0$ taking $\lambda_I=|\ell_I|/u_0$; If $\ell_I=0$, the corresponding
positive multiplier can be set to zero. Hence, the minimum is understood as an infimum when zero coefficients are present.

Let us denote by $f(u)$ the sum in the curly bracket on the right hand side of \eqref{eq:thm3.5prf}: 
  $f(u):={\kappa(\ell,\delta)}/{u}
  +\E_{P_0}[(1-{D}/{u})_+]$.
At values of $u$ with $P_0(D=u)=0$, differentiation of $f$ under expectation gives 
  $f'(u)=({h(u)-\kappa(\ell,\delta)})/{u^2}$.
If $\kappa(\ell,\delta)\ge\E[D]$, then $h(u)\le\kappa(\ell,\delta)$ for all $u$, so $f$ decreases to $1$ as $u\to\infty$. In the interior case, $h(u^\star)=\kappa(\ell,\delta)$ makes $u^\star$ stationary; since $h$ is nondecreasing, $f$ decreases before $u^\star$ and increases after it. Substituting  the moment equation gives
  $f(u^\star)
  = {\kappa(\ell,\delta)}/{u^\star}  + \E[
    (1-{D}/{u^\star})\1_{\{D\le u^\star\}} ] =P_0(D\le u^\star)$.
Finally, $D=[b-s_\ell(X)]_+$ implies $\{D\le u^\star\}=\{s_\ell(X)\ge b-u^\star\}$ up to the stated no-atom boundary convention.
\end{proof}

Unlike the robust expectation case, where $\kappa(\ell, \delta)$ is exactly the robust premium, here the worst-case probability of interest is rewritten in terms of the probability of a reasonably enlarged neighborhood of the original event $A$, evaluated in  the baseline model $P_0$. The enlargement is characterized by a threshold shift of $u^*$, which is again monotonically determined by $\kappa(\ell, \delta)$, as in \eqref{eq:thm3.5}.

\section{Budget-Aware Linear Signature Surrogates}
\label{sec:learning}

The closed-form part of Section~\ref{sec:theory} applies to affine signature scores and affine signature half-space events. For general path-dependent functionals, the approximation property of signatures in Proposition~\ref{thm:signature-universality} suggests a natural route: replace a target functional $F:\cX\to\R$ by an {\it affine score surrogate} $s_{\ell,\beta}$, that is,  
\begin{equation} \label{eq:affinesurrogate}
  F(x) \approx s_{\ell,\beta}(x)=\beta+\sum_{I\in\cW}\ell_I\phi_I(x), \quad x \in \mathcal X 
\end{equation}
and then apply the affine Sig-OT formulas in section \ref{sec:theory} to this surrogate $s_{\ell,\beta}$ in \eqref{eq:affinesurrogate}.

It is important to note approximation quality in \eqref{eq:affinesurrogate} alone is not enough. In the affine formulas, the size of the robust correction is controlled by   $\kappa(\ell,\delta)=\sum_{I\in\cW}|\ell_I|\delta_I$.
For expectations \eqref{eq:prop3.3} in Proposition \ref{prop:affine-expectation}, this is the relaxed robust premium, while for half-space probabilities \eqref{eq:thm3.5}, it determines the shifted threshold in Theorem~\ref{thm:scalar-reduction}. A surrogate that fits $F$ well but uses many high-budget coordinates may therefore produce an uninformative robust upper bound. Thus, the surrogate-fitting problem in \eqref{eq:affinesurrogate} is not an ordinary regression; it is the problem of approximating $F$ while keeping  $\kappa(\ell, \delta)$ small.

Now fix a selected word set $\cW$ and take a calibrated budget vector $\widehat\delta=(\widehat\delta_I)_{I\in\cW}$ as given. Given baseline training data $(x_i,F(x_i))_{i=1}^n$, we fit the affine score $s_{\ell, \beta}$ to $F$ by the budget-weighted {\it least absolute shrinkage and selection operator} (lasso), that is, finding $\widehat{\ell}, \widehat{\beta}$ by minimizing the following function: 
\begin{equation} \label{eq:LASSO}
  (\widehat\ell,\widehat\beta)
  \in
  \argmin_{\ell,\beta}
  \Big\{
    \frac1n\sum_{i=1}^n
    \big(F(x_i)-s_{\ell,\beta}(x_i)\big)^2
    +
    \alpha\sum_{I\in\cW}\widehat\delta_I|\ell_I|
  \Big\}.
\end{equation}
The penalty is the empirical effective budget
  $\widehat\kappa(\ell)
  =
  \sum_{I\in\cW}\widehat\delta_I |\ell_I|$.
The tuning parameter $\alpha\ge0$ mediates the tradeoff between the baseline approximation and downstream robustness: $\alpha=0$ gives ordinary least squares over the selected signature coordinates, while larger $\alpha$ favors scores with smaller certified sensitivity.  The fitted affine surrogate $\widehat{s}:= s_{\widehat\ell,\widehat\beta}(x)$, $x \in \mathcal X$, and the effective budget $\widehat{\kappa}:=\kappa(\widehat\ell,\widehat\delta)$ are used with the results in section~\ref{sec:theory}.

\begin{remark}
    In certain learned applications, it is beneficial to apply an optional scalar calibration in addition to the affine score surrogate to improve the accuracy of approximation of $F$. Let $g:\R\to\R$ be a monotone Lipschitz map with Lipschitz constant $L_g$ and set $\widehat G(x)=g(\widehat s(x))$, $x \in \mathcal X$. Then by 
  $|\widehat G(x)-\widehat G(y)|
  \le
  L_g\sum_{I\in\cW}|\widehat\ell_I|\,|\phi_I(x)-\phi_I(y)|$,
the expectation correction is at most $L_g\widehat\kappa$. In addition, if $g$ is nondecreasing, an event $\{\widehat G\ge b\}$ reduces to a threshold event in the affine score by setting
  $\tau_b:=\inf\{t\in\R:g(t)\ge b\}$. Thus, such a monotone $1$-Lipschitz {\it postlink} can improve scalar calibration without changing the coefficient vector that defines the effective budget. In the following examples, it is used only as a post-fitting refinement; the central budget-aware object remains the affine signature score $\widehat s$.
\end{remark}

\section{Examples and Numerical Experiments}
\label{sec:numerics}

This section presents examples and numerical experiments that illustrate the practical use of the Sig-OT framework. Before turning to the individual examples, we first describe the common workflow for computing a robust expectation or probability in this framework. The workflow has three stages.
\begin{enumerate}[label=(\arabic*)]
  \item \textbf{Budget calibration.} Fix a baseline model, and estimate a budget vector $\widehat\delta$ through a stress model.
  \item \textbf{Score specification.} Choose the sparse affine signature score $s_{\widehat\ell, \widehat\beta}$, either manually or by fitting a general functional with the effective-budget weighted lasso from Section \ref{sec:learning}. Compute the robust premium $\kappa(\widehat\ell,\widehat\delta)$.
  \item \textbf{Robust evaluation.} Evaluate the relaxed robust expectation correction (cf. Proposition \ref{prop:affine-expectation}) or shifted-threshold probability (cf. Theorem \ref{thm:scalar-reduction}) under the baseline law.
\end{enumerate}

We begin in section~\ref{subsec:empirical-budget-protocol} with comments on empirical budget calibration. The remaining experiments then were split into two classes. Subsection~\ref{subsec:linear-case-study} uses controlled sparse affine-signature scores. Subsection~\ref{subsec:applications} considers general path-dependent applications: first path-dependent option pricing for robust expectations, and then reserve-path risk events for robust probabilities.

Throughout the numerical examples, the sampled path is augmented by normalized time \(t/T\), while the state channels are kept in their original units before signature extraction. We compute geometric signatures of the piecewise-linear interpolation of the sampled augmented path using \texttt{iisignature} \cite{iisignature}. To place heterogeneous coordinates on a comparable scale, every signature coordinate is divided by a baseline-fitted robust scale
  $a_I
  =
  \max\!\{
  1.4826\,\operatorname{median}_{j}
  |S^I(x_j)-\operatorname{median}_{r}S^I(x_r)|,
  10^{-8}
  \}$.
The same baseline-fitted factors \(a_I\) are then used for the corresponding stress sample and for budget calibration. Thus, unless explicitly stated otherwise, \(\widetilde S^I=S^I/a_I\) denotes a signature coordinate normalized by median absolute deviation.

\subsection{On calibrating \texorpdfstring{$\delta$}{delta} and covering stress laws}
\label{subsec:empirical-budget-protocol}

Recall that the budget vector for an ambiguity set in Definition~\ref{def:sigot-ambiguity} has a joint interpretation, i.e., the coordinatewise coverage must be realized by the same coupling. Given a benchmark or stress law $P_1$, any coupling $\pi\in\Pi(P_1,P_0)$ induces the coordinate budgets
  $\delta_I^\pi
  :=
  \int |\phi_I(x)-\phi_I(y)|\,\pi(dx,dy),\,\forall  I\in\cW$,
and hence, certifies $P_1\in\cU_{\delta^\pi}(P_0)$. More generally, any coordinate-wise upper bound $\delta_I\ge\delta_I^\pi$ also covers $P_1$. Thus, the covering vector is not unique and it is a priori not clear which one is more suitable. For a fixed sparse affine score $s_{\ell,\beta}$, recall that the relevant robust premium scalar to be minimized is   $\kappa(\ell,\delta)=\sum_{I\in\cW}|\ell_I|\delta_I$, so a jointly certified natural calibration is to choose a coupling that is favorable in the active, high-weight coordinates. Let $\widehat P_0$ and $\widehat P_1$ be empirical laws of baseline samples $x_1,\ldots,x_n$ and stress samples $y_1,\ldots,y_m$, then the favorable joint coupling is thus the score-weighted one 
  $\widehat\pi^\ell
  \in
  \argmin_{\pi\in\Pi(\widehat P_1,\widehat P_0)}
  \sum_{I\in\cW}|\ell_I|
  \sum_{i,j}\pi_{ij}|\phi_I(y_i)-\phi_I(x_j)|$,
and the corresponding joint covering vector is estimated by
\begin{equation}\label{eq:joint delta}
      \widehat\delta_I^{\mathrm{joint}}:
  =
  \sum_{i,j}\widehat\pi^\ell_{ij}|\phi_I(y_i)-\phi_I(x_j)|.
\end{equation}
This construction is meaningful when the affine direction is fixed in advance. However, for learned path functionals, since the fitted coefficients already depend on the budget used in the weighted lasso, one needs to calibrate $\delta$ independently before knowing the sparse affine-signature surrogate. A practical way is to first use a plain linear regression or a plain lasso to obtain a preliminary affine-signature surrogate and then calibrate a $\widehat\delta_I^{\mathrm{joint}}$ in \eqref{eq:joint delta} with the preliminary affine-surrogate and then use the weighted lasso.

\paragraph{A decoupled proxy.}
Yet another computational difficulty is that common-coupling calibration is a multidimensional empirical OT problem. With $n$ baseline paths and $m$ stress paths, a Sinkhorn-type entropic solver typically forms an $n\times m$ cost or kernel matrix and performs work of $O(mn)$ order per scaling iteration, with comparable memory cost unless additional structure is exploited. This can be computationally expensive even with mini-batches. By contrast, one-dimensional empirical $W_1$ distances are much cheaper to compute coordinate by coordinate. This gives the decoupled proxy
\begin{equation}
  \widehat\delta_I^{\mathrm{dec}}:
  =
  W_1\big((\phi_I)_\#\widehat P_1,(\phi_I)_\#\widehat P_0\big),
  \qquad I\in\cW.
\end{equation}
Note that it is generally not a coverage certificate for $\widehat P_1\in\cU_{\widehat\delta}(\widehat P_0)$, because the coordinate-wise optimal couplings generally do not agree. Nevertheless, we notice that since we mostly work with sparse affine signature scores, such proxies are usually not too far from a certified budget vector $\widehat\delta_I^{\mathrm{joint}}$. Table \ref{tab:delta comparisons} compares the decoupled coordinate-wise proxy with the score-weighted common-coupling budget \footnote{The joint vector is estimated from \(16\) subsampled empirical OT problems per replicate, each using \(512\) baseline and \(512\) stress paths. Each subproblem is solved by entropic Sinkhorn with regularization \(10^{-2}\), at most \(400\) iterations, and tolerance \(10^{-9}\); the reported coordinate costs are averaged across the subsampled couplings.} for a fixed sparse affine signature score 
  $s(X)
  =
  -0.35\,\widetilde S^X
  +0.15\,\widetilde S^{tX}
  -0.45\,\widetilde S^{Xt}
  +0.20\,\widetilde S^{XX}$,
that we will further consider in Subsection \ref{subsec:linear-case-study}. Here, the baseline model is the log-price of a geometric Brownian motion with $\mu=0.05$ and $\sigma=0.2$, and the stress models are as in \eqref{eq:stress models}. The selected signature coordinates are normalized to a comparable scale before the comparison. We report both the aggregate budget size and the induced effective budget $\kappa$ with the same calibration convention as in the controlled linear expectation experiment below. In Table \ref{tab:delta comparisons} of the sparse affine linear scores (surrogates), the decoupled proxy preserves the scale of the jointly certified calibration rather well.

\begin{table}[ht]
\centering
\scriptsize
\caption{Decoupled versus score-weighted common-coupling Sig-OT budgets and induced effective budgets.}
\label{tab:delta comparisons}
\begin{tabular}{lrrrr}
\toprule
Stress model & $\|\widehat\delta^{\mathrm{dec}}\|_1$ & $\|\widehat\delta^{\mathrm{joint}}\|_1$ & $\widehat\kappa^{\mathrm{dec}}$ & $\widehat\kappa^{\mathrm{joint}}$ \\
\midrule
Adverse drift & 2.593 & 2.577 & 0.762 & 0.759 \\
Elevated volatility & 7.208 & 7.261 & 1.733 & 1.748 \\
Negative jumps & 0.835 & 0.962 & 0.207 & 0.241 \\
Crisis mixture & 1.231 & 1.301 & 0.301 & 0.315 \\
\bottomrule
\end{tabular}
\par\smallskip
\parbox{0.96\textwidth}{\scriptsize Based on \(20{,}000\) baseline and \(20{,}000\) stress paths on \(2{,}000\) time steps, signature level \(k=2\), baseline-fitted MAD scaling, and five independent seed replicates. Entries are replicate means.}
\end{table}

\paragraph{On the coverage of the ambiguity set.}
Notice that a budget vector $\delta$ always exists, as long as the selected signature coordinates have finite first moments under $P_0$ and $P_1$. This suggests that such ambiguity sets cover heterogeneous model uncertainties. We examine how empirical budget changes across parameter sweeps and compare heterogeneous models in Appendix~\ref{subsec:app-heterogeneous-coverage}. In addition, Appendix~\ref{app:analytic-bounds} gives two preliminary model-based upper bounds for feasible jointly certified budgets, based on synchronous diffusion couplings and fixed-grid Gaussian couplings, without any empirical constructions.

\subsection{Controlled sparse affine-signature robust values}
\label{subsec:linear-case-study}

We consider stress tests for two fixed sparse affine functionals of the time-augmented log-price signature, inducing a risk factor as well as a market tail event, under a controlled log-price model library. The observed state is the logarithmic return $X_t=\log(S_t/S_0)$, with $X_0=0$, simulated over $[0,1]$ with $T = 1$. The baseline law $P_0$ for the price of the stock is induced by the geometric Brownian motion GBM(0.05, 0.20), i.e. $\mathrm dS_t=S_t(\mu \,\mathrm dt+\sigma \,\mathrm dW_t)$  
  with $\mu=0.05$, $\sigma=0.20$. Equivalently, the log-price path $X$ satisfies
  $\mathrm d X_t=(\mu-\sigma^2/2)\,\mathrm dt+\sigma\,\mathrm dW_t$. 

The stress panel contains four deliberately separated alternatives:
\begin{equation}\label{eq:stress models}
    \begin{array}{ll}
\text{adverse-drift GBM} & \mu=-0.10,\ \sigma=0.20,\\
\text{elevated-volatility GBM} & \mu=0.05,\ \sigma=0.40,\\
\text{negative-jump Merton model} & \mu=0.05,\ \sigma=0.18,\ \lambda=1.0,\ m_J=-0.12,\ s_J=0.08,\\
\text{crisis-regime mixture} &
0.90\,\mathrm{GBM}(0.05,0.20)+0.10\,\mathrm{GBM}(-0.20,0.40).
\end{array}
\end{equation}

Here, $\mathrm{GBM}(\mu,\sigma)$ denotes the stock-price GBM above, represented through its log-price path, the Merton specification uses Gaussian jumps $J\sim N(m_J,s_J^2)$ with the usual exponential compensator,\footnote{For the discontinuous Merton paths, we record the log price on the uniform simulation grid and compute the ordinary geometric signature of its piecewise-linear interpolation.} and the mixture draws one latent regime for each whole path. In other words, the panel contains four common types of model uncertainties in market modeling.

The purpose of the examples below is twofold. First, although the affine-signature scores that we will consider are rather simple, they show that sparse affine signature scores can indeed induce meaningful risk factors, as well as tail events, that are sensitive to different types of model uncertainties. Second, they illustrate how tight the Sig-OT-induced robust values can be. Throughout, we write \(\widetilde S^I=\phi_I\) for the baseline-fitted MAD-normalized signature coordinate defined at the beginning of this section.

\subsubsection{Robust expectation for a sparse linear signature risk factor}
\label{subsubsec:linear-expectation}

Consider the following sparse affine signature market risk factor:
\begin{equation}
\begin{aligned}
     s_{\mathrm{risk}}(X)&:=-0.35\,\widetilde S^X
  +0.15\,\widetilde S^{tX}
  -0.45\,\widetilde S^{Xt}
  +0.20\,\widetilde S^{XX}\\
  &=
  -0.35\,\widetilde S^X
  -0.30\,\widetilde S^{Xt}
  +0.15\,\widetilde A^{tX}
  +0.20\,\widetilde S^{XX},
\end{aligned}
\end{equation}
where \(\widetilde A^{tX}:=\widetilde S^{tX}-\widetilde S^{Xt}\) denotes the normalized time-state L\'evy area, following the notation in Example \ref{ex:meaning of coordinates}.
This decomposition has the following transparent interpretation. The term
\(-\widetilde S^X\) penalizes terminal losses, while
\(-\widetilde S^{Xt}\) penalizes paths that remain low on average over the horizon. The positive L\'evy-area term \(+0.15\,\widetilde A^{tX}\) penalizes losses that occur early in the horizon or persist before a subsequent recovery. Conditional on the terminal return, \(\widetilde A^{tX}\) is larger for paths that spend more time at low levels than for comparable late deteriorations. Finally, \(\widetilde S^{XX}\) is a second-order terminal-magnitude feature and therefore adds sensitivity to large terminal losses. Thus, we interpret this score as a downside-oriented path risk factor.

Here, the selected signature coordinates are $\mathcal{W}:=\{X, tX, Xt, XX\}$. For each stress model listed in \eqref{eq:stress models}, we first calibrate the joint budget $ \widehat\delta_I^{\mathrm{joint}}$ for all $I\in \mathcal{W}$ as in \eqref{eq:joint delta} to cover the stress model. The relaxed Sig-OT robust premium induced by $\widehat\delta_I^{\mathrm{joint}}$ is then $\widehat\kappa_{\mathrm{joint}} = \sum_{I\in\cW}|\ell_I|\widehat\delta_I^{\mathrm{joint}}$ for each corresponding stress model. For each stress model, Table~\ref{tab:linear-risk-factor-expectation} reports baseline values, stress values, the value shifts, the relaxed Sig-OT relaxed premium, and the corresponding tightness ratios.

\begin{table}[ht]
\centering
\scriptsize
\caption{Robust expectation for the sparse affine signature market risk factor.}
\label{tab:linear-risk-factor-expectation}
\begin{tabular}{lrrrrr}
\toprule
Stress model & Base & Stress & Shift & $\widehat\kappa$ & Shift/$\widehat\kappa$ \\
\midrule
Adverse drift & 0.230 & 0.794 & 0.564 & 0.759 & 0.743 \\
Elevated volatility & 0.230 & 1.354 & 1.124 & 1.748 & 0.643 \\
Negative jumps & 0.230 & 0.348 & 0.119 & 0.241 & 0.493 \\
Crisis mixture & 0.230 & 0.479 & 0.249 & 0.315 & 0.792 \\
\bottomrule
\end{tabular}
\par\smallskip
\parbox{0.96\textwidth}{Based on five replicates; \(20{,}000\) paths per law for budget calibration and \(500{,}000\) paths per law for expectation evaluation; \(2{,}000\) time steps; \(k=2\); baseline MAD scaling. The joint coupling uses \(16\) Sinkhorn batches of size \(512\times512\), with regularization \(10^{-2}\). Entries are replicate means.}
\end{table}

Table~\ref{tab:linear-risk-factor-expectation} is also consistent with the downside interpretation of the score. All four stress laws increase the mean score relative to the baseline, with elevated volatility producing the largest shift and effective budget. 

\subsubsection{Worst-case probability for a sparse linear signature tail score}
\label{subsubsec:linear-probability}

We next consider sparse affine signature-induced market tail events. The tail score we consider is
\begin{equation}
  s_{\mathrm{tail}}(X):
  =
  -0.25\,\widetilde S^{Xt}
  -0.20\,\widetilde S^{Xtt}
  +0.20\,\widetilde S^{XXt}
  -0.35\,\widetilde S^{XXXt}.
\end{equation}
As in the preceding risk-factor example, \(-\widetilde S^{Xt}\) penalizes paths that remain low on average over the horizon. The term
\(-\widetilde S^{Xtt}\) is a further time-weighted version of this effect: persistent or early negative moves receive larger weight through the trailing time integrals. The coordinates \(\widetilde S^{XXt}\) and \(-\widetilde S^{XXXt}\) then add magnitude and asymmetry: large moves affect the score, while large negative excursions are amplified by the cubic term. Hence, high values of \(s_{\mathrm{tail}}\) identify paths with persistent low levels and large negative excursions along the horizon. For a high threshold \(b\), the event \(A_b=\{s_{\mathrm{tail}}\ge b\}\) is therefore the upper tail of this downside path-score distribution, and we interpret it as a pathwise market-tail event.

Here, the selected signature coordinates are $\mathcal W=\{Xt,Xtt,XXt,XXXt\}$.
For each stress model in \eqref{eq:stress models}, we first calibrate the joint budget \(\widehat\delta_I^{\mathrm{joint}}\) for all \(I\in\mathcal W\), as in \eqref{eq:joint delta}, so that the stress law is covered by the Sig-OT ambiguity set. The induced effective budget is
$\widehat\kappa_{\mathrm{joint}}
  =
  \sum_{I\in\mathcal W}|\ell_I|\widehat\delta_I^{\mathrm{joint}}$.
For each threshold \(b\), we then apply Theorem~\ref{thm:scalar-reduction}. In writing $D_b=[b-s_{\mathrm{tail}}(X)]_+$ and $X\sim P_0$,
the shifted-threshold parameter \(u_b^\star\) is determined by
  $\E_{P_0}\![D_b\mathbf 1_{\{D_b\le u_b^\star\}}]
  =
  \widehat\kappa_{\mathrm{joint}}$, 
which is calculated by bisection in the empirical baseline sample. If $\widehat\kappa_{\mathrm{joint}}\ge\E[D_b]$, the relaxed robust probability bound is $1$. In the interior case $0<\widehat\kappa_{\mathrm{joint}}<\E[D_b]$, the relaxed robust probability is then evaluated under the baseline Monte Carlo sample as
  $P_0(s_{\mathrm{tail}}(X)\ge b-u_b^\star)$.

For each stress model, Table \ref{tab:linear-tail-probability} reports baseline probabilities, stress probabilities, the amplification ratios Stress/Base, as well as the Robust/Stress ratios for representative thresholds $b\in\{20,30,40\}$, spanning moderate to deeper rare-tail levels. Table~\ref{tab:linear-tail-probability} illustrates a different aspect of the method. For this particular tail score, the adverse-drift stress is not the main source of tail amplification, so we focus on the other three stress models, where the baseline probabilities are extremely small and the stress probabilities can be amplified by several orders of magnitude at deeper thresholds. As illustrated in the table, if the market follows one of these stress models while the baseline model is used without robustification, the tail event would be severely underestimated. With a Sig-OT ambiguity set, relaxed robust probabilities dominate the corresponding stress probabilities in all reported cases. Moreover, the robust probabilities remain close to the stress scale compared to the much larger baseline-to-stress amplification, providing an effective correction for tail underestimation.

\begin{table}[ht]
\centering
\scriptsize
\caption{Robust probabilities for sparse affine signature-induced market tail events.}
\label{tab:linear-tail-probability}
\begin{tabular}{lrrrrrr}
\toprule
Stress model & $b$ & Base prob. & Stress prob. & Robust prob. & Stress/Base & Robust/Stress \\
\midrule
Adverse drift & 20 & $1.715\times 10^{-4}$ & $1.939\times 10^{-3}$ & 0.069 & 11.322 & 35.819 \\
 & 30 & $1.560\times 10^{-5}$ & $2.754\times 10^{-4}$ & 0.044 & 17.903 & 160.034 \\
 & 40 & $2.200\times 10^{-6}$ & $4.640\times 10^{-5}$ & 0.032 & 22.441 & 692.245 \\
\midrule
Elevated volatility & 20 & $1.715\times 10^{-4}$ & 0.057 & 0.297 & 334.541 & 5.195 \\
 & 30 & $1.560\times 10^{-5}$ & 0.031 & 0.196 & $1.997\times 10^{3}$ & 6.418 \\
 & 40 & $2.200\times 10^{-6}$ & 0.018 & 0.146 & $8.438\times 10^{3}$ & 8.250 \\
\midrule
Negative jumps & 20 & $1.715\times 10^{-4}$ & $5.129\times 10^{-3}$ & 0.031 & 29.960 & 5.959 \\
 & 30 & $1.560\times 10^{-5}$ & $1.871\times 10^{-3}$ & 0.019 & 122.060 & 9.973 \\
 & 40 & $2.200\times 10^{-6}$ & $8.197\times 10^{-4}$ & 0.013 & 390.638 & 16.199 \\
\midrule
Crisis mixture & 20 & $1.715\times 10^{-4}$ & 0.016 & 0.053 & 92.308 & 3.366 \\
 & 30 & $1.560\times 10^{-5}$ & $9.640\times 10^{-3}$ & 0.033 & 629.523 & 3.463 \\
 & 40 & $2.200\times 10^{-6}$ & $6.257\times 10^{-3}$ & 0.024 & $2.976\times 10^{3}$ & 3.852 \\
\bottomrule
\end{tabular}
\par\smallskip
\parbox{0.96\textwidth}{Based on five replicates; \(20{,}000\) paths per law for budget calibration and \(3{,}000{,}000\) paths per law for probability evaluation; \(2{,}000\) time steps; \(k=4\); baseline MAD scaling. The joint coupling uses \(16\) Sinkhorn batches of size \(512\times512\), with regularization \(10^{-2}\). Entries are replicate means.}
\end{table}

\subsection{Path-dependent model-risk applications}
\label{subsec:applications}

We move from controlled signature scores to financially motivated model-risk calculations with general path-dependent functionals/events. We  consider path-dependent option pricing with uncertain volatility in section \ref{subsubsec:option-pricing}, and then a reserve-path ruin probability problem under model misspecification in section \ref{subsubsec:reserve-risk}.

\subsubsection{Robust expectation: path-dependent option pricing}
\label{subsubsec:option-pricing}

We first consider robust pricing of path-dependent options as an expectation-side learned-surrogate illustration. The sensitivity of exotic and path-dependent derivatives to model choice is a classical source of model risk: two models can be close to or even calibrated to the same liquid vanilla option prices while assigning materially different values to claims depending on averages, extrema, or barrier crossings. This issue is important in several different areas of mathematical finance \cite{cont2006modeluncertainty}, \cite{hobson1998lookback,brownhobsonrogers2001barrier}, \cite{coxkallblad2017asian}.

We here focus on two simple but genuinely path-dependent claims: an arithmetic Asian call and a lookback call. The baseline is a risk-neutral GBM with $S_0=100$, maturity $T=1$, rate $r=0.02$, dividend yield $q=0$, and baseline volatility $\sigma_0=0.20$. The stress laws are volatility-shifted risk-neutral GBMs with
\begin{equation}
  \sigma\in\{0.10,0.15,0.25,0.30\},
  \qquad \mathrm dS_t=(r-q)S_t\,\mathrm dt+\sigma S_t\, \mathrm dW_t, \, \, t \ge 0.  
\end{equation}
We use the time-augmented log-price representation $(t,X_t)$, where $X_t=\log(S_t/S_0)$. On a uniform monitoring grid \(0=t_0<t_1<\cdots<t_m=T\), the payoff families are
\begin{equation}
  G^{\mathrm{Asian}}_K(S)
  =e^{-rT}\Big(\frac{1}{m}\sum_{i=1}^m S_{t_i}-K\Big)_+, \quad 
  G^{\mathrm{Lookback}}_K(S)
  =e^{-rT}\Big(\max_{0\le i\le m}S_{t_i}-K\Big)_+.
\end{equation}
Thus, the Asian payoff is the right-endpoint Riemann-grid version of the continuously monitored arithmetic average, while the lookback payoff is the grid-monitored running maximum. We report all experiments for \(K\in\{95,100,105\}\), corresponding to in-the-money, at-the-money, and out-of-the-money regimes around \(S_0=100\). We use the time-augmented signature of $(t,X_t)$ up to the truncation level $4$, i.e., $\mathcal{W}:=\{I, |I|\leq 4\}$.

For each payoff and each stress volatility, we recall the calibration-and-fitting workflow. First, in baseline training samples, we fit a preliminary sparse affine signature score by an unweighted lasso. Second, using absolute preliminary coefficients as task weights, we solve the empirical common-coupling OT problem between baseline and stress signature samples (cf. Subsection \ref{subsec:empirical-budget-protocol}), and obtain a coordinate-wise budget vector $\widehat\delta^{\mathrm{joint}}$ certified by one coupling. Third, we refit the payoff on the baseline samples using the effective-budget weighted lasso
  $n^{-1}\sum_{i=1}^n
  (G_K(x_i)-s_{\ell,\beta}(x_i))^2
  + \alpha\sum_{I\in\cW}\widehat\delta_I^{\mathrm{joint}}|\ell_I|$, 
with a manually chosen $\alpha=1$. We then apply a $1$-Lipschitz spline postlink $g$, as described in Section~\ref{sec:learning}, to improve the precision of the adjustment. For the fitted surrogate $\widehat G=g(\widehat s)$, the relaxed robust surrogate upper price is 
  $\E_{P_0}[\widehat G(X)]+\widehat\kappa$ with  
  $\widehat\kappa
  =  \sum_{I\in\cW}|\widehat\ell_I|\,\widehat\delta_I^{\mathrm{joint}}$. 
Table~\ref{tab:path-option-vol-sweep} reports the baseline price, the price of the stress model, the shift, the surrogate-learned robust value premium $\widehat\kappa$, the tightness ratio, and the surrogate fit $R^2$ for each stress model and each strike price tested.

\begin{table}[ht]
\centering
\scriptsize
\caption{Sig-OT robust pricing diagnostics for path-dependent option pricing.}
\label{tab:path-option-vol-sweep}
\resizebox{\textwidth}{!}{%
\begin{tabular}{rrrrrrrrrrrrrrr}
\toprule
\multicolumn{3}{c}{Stress and strike} &
\multicolumn{6}{c}{Asian call} &
\multicolumn{6}{c}{Lookback call} \\
\cmidrule(lr){1-3}\cmidrule(lr){4-9}\cmidrule(lr){10-15}
$\sigma$ & $\Delta\sigma$ & $K$ &
Base & Stress & Shift & $\widehat\kappa$ & $|\mathrm{Shift}|/\widehat\kappa$ & Fit $R^2$ &
Base & Stress & Shift & $\widehat\kappa$ & $|\mathrm{Shift}|/\widehat\kappa$ & Fit $R^2$ \\
\midrule
0.10 & -0.10 & 95 & 7.962 & 6.291 & -1.672 & 5.070 & 0.330 & 0.980 & 22.428 & 13.949 & -8.479 & 9.832 & 0.863 & 0.960 \\
 &  & 100 & 5.054 & 2.798 & -2.256 & 4.575 & 0.493 & 0.967 & 17.527 & 9.048 & -8.479 & 9.832 & 0.863 & 0.960 \\
 &  & 105 & 2.979 & 0.885 & -2.094 & 3.784 & 0.554 & 0.943 & 13.156 & 5.004 & -8.152 & 9.595 & 0.850 & 0.957 \\
\midrule
0.15 & -0.05 & 95 & 7.962 & 7.053 & -0.909 & 2.873 & 0.318 & 0.985 & 22.428 & 18.133 & -4.295 & 5.912 & 0.728 & 0.963 \\
 &  & 100 & 5.054 & 3.926 & -1.129 & 2.623 & 0.432 & 0.975 & 17.527 & 13.232 & -4.295 & 5.912 & 0.728 & 0.963 \\
 &  & 105 & 2.979 & 1.895 & -1.084 & 2.101 & 0.520 & 0.957 & 13.156 & 8.975 & -4.181 & 5.858 & 0.715 & 0.962 \\
\midrule
0.25 & 0.05 & 95 & 7.962 & 8.933 & 0.971 & 2.568 & 0.382 & 0.986 & 22.428 & 26.858 & 4.429 & 5.801 & 0.767 & 0.961 \\
 &  & 100 & 5.054 & 6.172 & 1.118 & 1.990 & 0.570 & 0.975 & 17.527 & 21.957 & 4.429 & 5.801 & 0.767 & 0.961 \\
 &  & 105 & 2.979 & 4.085 & 1.105 & 1.464 & 0.760 & 0.958 & 13.156 & 17.512 & 4.356 & 5.603 & 0.781 & 0.959 \\
\midrule
0.30 & 0.10 & 95 & 7.962 & 9.953 & 1.991 & 5.913 & 0.338 & 0.974 & 22.428 & 31.422 & 8.994 & 11.420 & 0.790 & 0.948 \\
 &  & 100 & 5.054 & 7.307 & 2.252 & 5.251 & 0.432 & 0.955 & 17.527 & 26.521 & 8.994 & 11.420 & 0.790 & 0.948 \\
 &  & 105 & 2.979 & 5.221 & 2.242 & 3.194 & 0.714 & 0.932 & 13.156 & 22.028 & 8.872 & 11.028 & 0.808 & 0.944 \\
\bottomrule
\end{tabular}
}
\par\smallskip
\parbox{0.98\textwidth}{Based on five replicates; \(20{,}000\) baseline and \(20{,}000\) stress paths for budget calibration, \(20{,}000\) baseline paths for fitting with \(16{,}000\) training observations, and \(500{,}000\) paths per law for price evaluation; \(2{,}000\) time steps; \(k=4\); \(\alpha=1\); baseline MAD scaling. The task-weighted joint OT stage uses Sinkhorn regularization \(10^{-2}\), batch size \(512\), and eight subsampled batches per outer iteration.}
\end{table}

Table~\ref{tab:path-option-vol-sweep} shows that the learned signature surrogates give stable pricing diagnostics for these smooth monitored payoffs. The lookback payoffs have consistently tighter ratios, roughly \(0.72\)--\(0.86\), while the Asian ratios range from about \(0.32\) to \(0.76\). Together with the high out-of-sample \(R^2\) values, this suggests that the budget-aware signature surrogate captures the main volatility sensitivity of these smooth path-dependent claims, although the robust premium remains looser for some Asian regimes with the chosen hyperparameters.

\begin{remark}[On severe path-dependent option price model risk]
\label{rem:Cont}

Notice that in the above examples, the option payoffs are rather smooth and stable; thus they can be relatively well approximated by sparse affine signature surrogates. It would be technically difficult to use them to learn more complicated exotic payoffs, e.g., discontinuous barrier-style payoffs whose value may hinge on a narrow crossing event. Indeed, such irregular payoffs might in general be more sensitive to model misspecification. An example was given in Example 4.4 of \cite{cont2006modeluncertainty}, where a jump-diffusion model and a calibrated local-volatility diffusion agree on a family of short-maturity call prices. Nevertheless, they generate very different path scenarios: one has discontinuous jumps, while the other is continuous and uses high short-term local volatility to reproduce the same smile. For a short-maturity knock-out call with strike \(K=105\) and barrier \(B=110\), the reported prices are \(2.73\) under the local-volatility model and \(1.63\) under the jump model, an economically large gap, approximately 40\% of the local-volatility price.

The reason is structural: in models with jumps, the asset can cross the barrier discontinuously, so the knock-out event has a quite different likelihood than under a continuous local-volatility diffusion. Although our current framework does not directly solve such problems due to the limited expressiveness of sparse affine signature scores, it yet yields two potentially interesting directions. First, instead of exactly fitting the payoff functional, is it possible to use sparse affine signature scores to construct risk factors that are sensitive to different types of events/functionals in a statistical sense? Second, except for signature coordinates of the log prices, are there other ways to factorize the path measure discrepancy to match specific type of events/functional sensitivities? These questions are left for future research. 
\end{remark}

\subsubsection{Robust probabilities: reserve-path risk events}
\label{subsubsec:reserve-risk}
We next consider a robust probability case study for insurance reserve paths. We consider the celebrated Cram\'er--Lundberg model (e.g., \cite{asmussenalbrecher2010}), where a compound Poisson process is used to model an insurance risk/reserve process. With initial capital $u\ge0$, the surplus process is defined by 
\begin{equation}
  U_t^u
  :=
  u+(1+\eta)\nu m_1t-\sum_{i=1}^{N_t}Y_i
  =:
  u+R_t,
  \qquad
  R_0=0,
\end{equation}
where $N_t$ is a Poisson process with intensity $\nu$, the claim sizes $Y_i$ are i.i.d. with first two moments $m_1,m_2$, and $\eta>0$ is the safety load. One of the important problems in risk theory is to calculate the tail probability in relation to the process $R_t$. For example, the event
  $\{\inf_{0\le t\le T}U_t^u\le0 \}
  = \{\sup_{0\le t\le T}(-R_t)\ge u\}$
represents that the insurance company goes into bankruptcy before a specified duration $T$. However, even for such a simple event, existing results in the literature do not admit simple methods for the
computation of its probability (e.g., \cite{asmussenalbrecher2010}).  In addition, if historical data are not adequately available to choose an appropriate distribution for claim sizes, as is typically the case in an exotic insurance situation, it is common to use a diffusion approximation
  $R_t^B
  =
  \eta\nu m_1t-\sqrt{\nu m_2}\,B_t$,
which matches the drift and variance of the aggregate claim process and is much easier to sample and analyze. 

However, this approximation creates a natural model-risk issue. The Cram\'er--Lundberg path and the Brownian path have quite different mechanisms for producing a large adverse reserve movement: under the Cram\'er--Lundberg model, a single large claim can cause an immediate downward jump, whereas the Brownian diffusion requires  continuous downward excursions against the positive safety-loading drift. When claim sizes are heavy-tailed, this difference can substantially affect probabilities of first-passage, low-extrema, early-shortfall, and drawdown. A Brownian baseline may therefore be a reasonable tractable surrogate for typical fluctuations, while still hugely underestimating the tail probabilities relevant for capital assessment.

Notably, many relevant reserve-risk events can naturally be written as functional half-space events
  $A_u(F)=\{F(R)\ge u\}$, for some proper $F$, such that Section~\ref{subsec:robust-probability} can be used, once $F$ is represented by a sparse affine signature score. The family of capital-risk scores that we consider is listed in Table \ref{tab:reserve-risk-targets}. 

\begin{table}[ht]
\centering
\small
\caption{Some typical reserve-path tail functionals.}
\label{tab:reserve-risk-targets}
\begin{tabular}{
>{\raggedright\arraybackslash}p{0.20\textwidth}
>{\raggedright\arraybackslash}p{0.43\textwidth}
>{\raggedright\arraybackslash}p{0.29\textwidth}}
\toprule
Target & Score $F(R)$ & Interpretation \\
\midrule
Constant boundary &
$\displaystyle \sup_{0\le t\le T}(-R_t)$ &
standard ruin event. \\
Sinusoidal vulnerability &
$\displaystyle F_{\mathrm{sin}}(R)=
\sup_{0\le t\le T}\left\{-R_t+a\sin^2(\pi t/T)\right\},\ a>0$ &
mid-horizon weighted shortfall. \\
Early shortfall &
$\displaystyle F_{\mathrm{early}}(R)=
\sup_{0\le t\le T}\bigl(1+\rho(1-t/T)\bigr)(-R_t)$ &
early shortfall penalty. \\
Maximum drawdown &
$\displaystyle F_{\mathrm{dd}}(R)=
\sup_{0\le t\le T}\left(\max_{0\le s\le t}R_s-R_t\right)$ &
drop from running maximum. \\
\bottomrule
\end{tabular}
\end{table}

Here, we use the Brownian diffusion law as the baseline and the Cram\'er--Lundberg law as the stress law.\footnote{The compound-Poisson reserve path is sampled on the same uniform grid as the Brownian path. Claims arriving between grid points are assigned to the first subsequent grid point, and we compute the ordinary geometric signature of the piecewise-linear interpolation of the resulting sampled path.} In the numerical study below, $T=100$, $\nu=1$, $\eta=0.1$, and the claim-size distribution is Pareto with tail index $2.2$. In Table \ref{tab:reserve-risk-targets}, all four targets use the common extrema-aware lift 
  $(t,R_t,\underline R_t,\overline R_t)$, where 
  $\underline R_t:=\min_{0\le s\le t}R_s$, 
  $\overline R_t:=\max_{0\le s\le t}R_s$,
so that the comparison is made under the same reusable Sig-OT geometry.

For each target functional $F$, we recall the calibration-and-fitting procedure in Section \ref{sec:learning}. First, in Brownian baseline samples, we fit a preliminary sparse affine signature surrogate for $F$. Second, using the preliminary coefficients as task weights, we calibrate a coordinate-wise Sig-OT budget vector $\widehat\delta$ between Brownian and Cram\'er--Lundberg samples on the lifted representation. Third, we refit the surrogate in Brownian samples using the effective-budget weighted lasso from Section~\ref{sec:learning}. Finally, we apply the monotone $1$-Lipschitz postlink and compute the effective budget
  $\widehat\kappa
  =
  \sum_{I\in\cW}|\widehat\ell_I|\,\widehat\delta_I$.
The robust surrogate probability is then obtained from the shifted-threshold formula in Theorem~\ref{thm:scalar-reduction} and evaluated under the Brownian baseline law.

The main reserve-risk table reports representative thresholds $u\in\{100,120,140,160,180,200\}$ across the constant-boundary, sinusoidal-vulnerability, early-shortfall, and drawdown targets. Alongside the Brownian baseline and Cram\'er--Lundberg Monte Carlo probabilities, we report the relaxed Sig-OT robust surrogate probability, the underestimating ratio Brownian/CL, the diagnostic ratio Sig-OT/CL, the effective budget, the induced threshold shift, and the surrogate fit quality.

\begin{table}[ht]
\centering
\scriptsize
\caption{Robust probabilities for reserve-path risk events.}
\label{tab:reserve-risk-family}
\resizebox{\textwidth}{!}{%
\begin{tabular}{lrrrrrrrrr}
\toprule
Target & $u$ & Brownian & CL-MC & Sig-OT & Brownian/CL & Sig-OT/CL & $\widehat\kappa$ & Shift & Fit $R^2$ \\
\midrule
Constant boundary & 100 & $4.025\times 10^{-4}$ & $4.791\times 10^{-3}$ & 0.089 & 0.084 & 18.568 & 4.151 & 59.131 & 0.996 \\
 & 120 & $3.220\times 10^{-5}$ & $3.028\times 10^{-3}$ & 0.066 & 0.011 & 21.778 & 4.151 & 75.073 & 0.996 \\
 & 140 & $2.160\times 10^{-6}$ & $2.092\times 10^{-3}$ & 0.052 & 0.001 & 24.784 & 4.151 & 91.913 & 0.996 \\
 & 160 & $1.600\times 10^{-7}$ & $1.522\times 10^{-3}$ & 0.042 & $1.048\times 10^{-4}$ & 27.916 & 4.151 & 109.356 & 0.996 \\
 & 180 & $<4.0\times 10^{-8}$ & $1.155\times 10^{-3}$ & 0.036 & $<3.5\times 10^{-5}$ & 31.053 & 4.151 & 127.219 & 0.996 \\
 & 200 & $<4.0\times 10^{-8}$ & $9.032\times 10^{-4}$ & 0.031 & $<4.5\times 10^{-5}$ & 34.288 & 4.151 & 145.393 & 0.996 \\
\midrule
Sinusoidal vulnerability & 100 & $9.099\times 10^{-4}$ & $6.612\times 10^{-3}$ & 0.135 & 0.138 & 20.420 & 5.416 & 52.323 & 0.957 \\
 & 120 & $5.896\times 10^{-5}$ & $3.896\times 10^{-3}$ & 0.097 & 0.015 & 24.782 & 5.416 & 67.645 & 0.957 \\
 & 140 & $2.880\times 10^{-6}$ & $2.563\times 10^{-3}$ & 0.074 & 0.001 & 28.912 & 5.416 & 84.184 & 0.957 \\
 & 160 & $1.600\times 10^{-7}$ & $1.814\times 10^{-3}$ & 0.060 & $8.810\times 10^{-5}$ & 32.912 & 5.416 & 101.486 & 0.957 \\
 & 180 & $<4.0\times 10^{-8}$ & $1.346\times 10^{-3}$ & 0.050 & $<3.0\times 10^{-5}$ & 36.992 & 5.416 & 119.291 & 0.957 \\
 & 200 & $<4.0\times 10^{-8}$ & $1.035\times 10^{-3}$ & 0.043 & $<3.9\times 10^{-5}$ & 41.162 & 5.416 & 137.455 & 0.957 \\
\midrule
Early shortfall & 100 & $3.882\times 10^{-3}$ & 0.012 & 0.192 & 0.314 & 15.535 & 7.809 & 56.543 & 0.983 \\
 & 120 & $5.386\times 10^{-4}$ & $7.713\times 10^{-3}$ & 0.140 & 0.070 & 18.196 & 7.809 & 70.792 & 0.983 \\
 & 140 & $5.332\times 10^{-5}$ & $5.243\times 10^{-3}$ & 0.109 & 0.010 & 20.740 & 7.809 & 86.395 & 0.983 \\
 & 160 & $4.240\times 10^{-6}$ & $3.798\times 10^{-3}$ & 0.088 & 0.001 & 23.162 & 7.809 & 102.902 & 0.983 \\
 & 180 & $2.800\times 10^{-7}$ & $2.874\times 10^{-3}$ & 0.073 & $9.741\times 10^{-5}$ & 25.571 & 7.809 & 120.046 & 0.983 \\
 & 200 & $<4.0\times 10^{-8}$ & $2.256\times 10^{-3}$ & 0.063 & $<1.8\times 10^{-5}$ & 27.886 & 7.809 & 137.640 & 0.983 \\
\midrule
Maximum drawdown & 100 & $9.064\times 10^{-4}$ & $6.864\times 10^{-3}$ & 0.099 & 0.132 & 14.445 & 3.873 & 50.755 & 0.837 \\
 & 120 & $7.072\times 10^{-5}$ & $4.113\times 10^{-3}$ & 0.070 & 0.017 & 17.117 & 3.873 & 66.698 & 0.837 \\
 & 140 & $4.000\times 10^{-6}$ & $2.711\times 10^{-3}$ & 0.054 & 0.001 & 19.857 & 3.873 & 83.577 & 0.837 \\
 & 160 & $3.200\times 10^{-7}$ & $1.915\times 10^{-3}$ & 0.043 & $1.662\times 10^{-4}$ & 22.595 & 3.873 & 101.099 & 0.837 \\
 & 180 & $<4.0\times 10^{-8}$ & $1.423\times 10^{-3}$ & 0.036 & $<2.9\times 10^{-5}$ & 25.327 & 3.873 & 119.038 & 0.837 \\
 & 200 & $<4.0\times 10^{-8}$ & $1.089\times 10^{-3}$ & 0.031 & $<3.7\times 10^{-5}$ & 28.281 & 3.873 & 137.298 & 0.837 \\
\bottomrule
\end{tabular}
}
\par\smallskip
\parbox{0.98\textwidth}{Based on five replicates; \(20{,}000\) paths per law for budget calibration, \(20{,}000\) Brownian paths for fitting with \(16{,}000\) training observations, and \(5{,}000{,}000\) paths per law for probability evaluation; \(2{,}000\) time steps; \(k=4\); \(\alpha=3\); baseline MAD scaling. The task-weighted joint OT stage uses Sinkhorn regularization \(10^{-2}\), batch size \(512\), and eight subsampled batches per outer iteration. Across the five replicates, \(<4.0\times10^{-8}\) is the one-hit empirical resolution, meaning no hit among \(25{,}000{,}000\) pooled paths.}
\end{table}

As illustrated in Table~\ref{tab:reserve-risk-family}, the Brownian diffusion approximation misses the jump-driven reserve tail risk at high capital levels. As \(u\) increases, the Brownian/CL ratio rapidly collapses toward zero, and in several rows the Brownian Monte Carlo estimate records no tail events among the \(2.5\times 10^7\) pooled paths in the five replicates. Thus, the Sig-OT ambiguity set provides a uniform correction across the four reserve-risk targets, and the induced robust probabilities remain much closer to the true Cram\'er--Lundberg probabilities compared to the significant underestimation induced by the Brownian-based approximation.

\paragraph{Benchmark against the $J_1$-OT robust probability.}
For the ordinary constant-boundary ruin event, we separately report the Skorokhod $J_1$ benchmark from Example 1 in \cite{blanchetmurthy2019}. This benchmark is a useful reference point because the transport cost therein is strongly related to the same crossing functional that defines ordinary ruin, and thus should be expected to be sharp within the class of OT-DRO stress tests for the same task. In their construction, the ambiguity set is a Wasserstein ball in the path space whose transport cost is induced by the Skorokhod $J_1$ metric; In Appendix~\ref{app:j1-ruin-benchmark}, we recall that approach and explain why it should be expected to be sharp.

In summary, the sharpness essentially comes from the distance-to-event identity below. For the constant crossing set
  $A_u=\{z:\sup_{0\le t\le T}z_t\ge u\}$,
one has
  $c_{J_1}(z,A_u)
  =
  (u-\sup_{0\le t\le T}z_t)_+^p$.
Thus, the $J_1$ enlargement of the event depends only on the same functional \(\sup_t z_t\) that appears in the event itself. Equivalently, the ambiguity radius is used only to close the vertical gap between the running maximum and the ruin level \(u\). It is not spent on any other feature of the path. Thus, as an OT stress test for the ordinary constant-boundary ruin probability, this benchmark has essentially no additional slack.

Table~\ref{tab:constant-ruin-tailored-comparison} compares this $J_1$ benchmark with two Sig-OT representations on the original BM19 reserve grid $u\in\{50,100,150,200,250\}$: the plain time-augmented reserve path $(t,R_t)$ and the extrema-aware lift $(t,R_t,\underline R_t,\overline R_t)$. Compared with the $J_1$-induced robust values, the robust values induced by Sig-OT are looser, especially for deeper tail levels, but still, in general, much better than the underestimation by the Brownian model misspecification. However, because of the specific sensitivity of $J_1$-distance to the supremum, one should not expect it to be as sensitive to other general path-dependent events and provide as informative robust bounds for them. Nevertheless, our approach creates looser but informative robust values for different path-dependent tail risks under the same reusable ambiguity set. In addition, it is expected that adding the maximum channel is materially improving the tightness.

\begin{table}[ht]
\centering
\scriptsize
\caption{Constant-boundary ruin diagnostic: $J_1$ versus plain and extrema-aware Sig-OT.}
\label{tab:constant-ruin-tailored-comparison}
\resizebox{\textwidth}{!}{%
\begin{tabular}{rrrrrrrrrr}
\toprule
$u$ & Brownian & CL-MC & $J_1$ & Sig-OT plain & Sig-OT ext. & Brownian/CL & $J_1$/CL & Sig-OT plain/$J_1$ & Sig-OT ext./$J_1$ \\
\midrule
50 & 0.052 & 0.031 & 0.229 & 0.360 & 0.291 & 1.688 & 7.475 & 1.571 & 1.268 \\
100 & $4.260\times 10^{-4}$ & $4.791\times 10^{-3}$ & 0.049 & 0.120 & 0.089 & 0.089 & 10.187 & 2.466 & 1.823 \\
150 & $4.360\times 10^{-7}$ & $1.771\times 10^{-3}$ & 0.018 & 0.067 & 0.047 & $2.462\times 10^{-4}$ & 10.392 & 3.640 & 2.540 \\
200 & $5.050\times 10^{-11}$ & $9.032\times 10^{-4}$ & 0.010 & 0.046 & 0.031 & $5.592\times 10^{-8}$ & 11.294 & 4.476 & 3.035 \\
250 & $6.750\times 10^{-16}$ & $5.433\times 10^{-4}$ & $6.500\times 10^{-3}$ & 0.034 & 0.023 & $1.243\times 10^{-12}$ & 11.966 & 5.297 & 3.533 \\
\bottomrule
\end{tabular}
}
\par\smallskip
\parbox{0.98\textwidth}{For each Sig-OT representation: based on five replicates; \(20{,}000\) paths per law for budget calibration, \(20{,}000\) Brownian paths for fitting with \(16{,}000\) training observations, and \(5{,}000{,}000\) paths per law for Monte Carlo evaluation; \(2{,}000\) time steps; \(k=4\); \(\alpha=3\); baseline MAD scaling; Sinkhorn regularization \(10^{-2}\), batch size \(512\), and eight subsampled batches per outer iteration. The Brownian and \(J_1\) columns are reproduced from Table 1 in Example 1 of \cite{blanchetmurthy2019}; all remaining columns are computed in the present study.}
\end{table}

\section{Conclusion and Further Work}
\label{sec:conclusion}

In this paper, we developed a signature-induced optimal-transport framework for path-dependent model risk. The main modeling choice is a coordinate-wise transport geometry on selected time-augmented signature coordinates, with all coordinate budgets certified by one common path coupling. This gives a reusable ambiguity geometry for path laws for general path-dependent payoffs/events.

On the theoretical side, we studied robust expectations and worst-case closed-event probabilities in this framework. We obtained tractable closed-form upper bounds on both quantities with linear signatures via an ambient space relaxation. Both robust values are governed by the same effective budget
$\kappa(\ell,\delta)=\sum_{I\in\cW}|\ell_I|\delta_I$,
which also motivates the budget-weighted lasso used to fit sparse signature surrogates for general path functionals.

We then presented two layers of numerical evidence for controlled sparse affine signature payoffs and tail events, as well as general path-dependent payoffs and events that arise from finance and insurance. In general, the robust values induced by the Sig-OT approach are informative but may be looser than those from functionally tailored ambiguity sets for specific tasks. However, the advantage is that our Sig-OT approach is a {\it reusable framework} across different types of payoffs and events.

As mentioned in Remark \ref{rem:Cont}, sparse affine signature scores are not expected to provide direct approximations of general irregular or implicit path functionals from finite samples, which remains an interesting challenge of our current approach. A promising direction is to use such scores as {\it interpretable pathwise risk factors} to be combined with domain-specific robust analysis, rather than solely as direct approximations of a target functional. In addition, it would also be interesting to extend the idea of factorizing path-space model uncertainty through signature coordinates beyond the present optimal-transport formulation and to connect it with machine-learning architectures for implicit or data-driven path functionals. Finally, the computational cost of high-dimensional empirical OT motivates stronger model-based results to select the budget vector \(\delta\), complementing the calibration of the stress-test and the preliminary limits in the Appendix~\ref{app:analytic-bounds}.

\paragraph{Code availability.}
The code used to produce the numerical results in this paper is available from the authors upon reasonable request.

\appendix

\section{Proof of Theorem~\ref{thm:finite-budget-otdro}}
\label{app:proofs}

The proof of Theorem~\ref{thm:finite-budget-otdro} uses the measurable-selection idea underlying Proposition~4 of \cite{blanchetmurthy2019}, together with a finite-dimensional supporting-hyperplane argument for the vector of transport budgets.

Write  $\Gamma(P_0):=\{\pi\in\cP(E\times E):\pi_2=P_0\}$, where $\pi(\mathrm dz,\mathrm dx)$ couples the adversarial variable $z$ with the baseline variable $x$. For $\pi\in\Gamma(P_0)$, set 
  $F(\pi):=\int f(z)\,\pi(\mathrm dz,\mathrm dx)$,
  $C_j(\pi):=\int c_j(z,x)\,\pi(\mathrm dz,\mathrm dx)$. Note that as in \cite{blanchetmurthy2019}, extended expectations in the primal problem are understood with the convention $+\infty-\infty:=+\infty$. Let us define 
\begin{equation}
  \Gamma_{\mathrm{fin}}(P_0)
  :=
  \Big\{
    \pi\in\Gamma(P_0):
    \int |f(z)|\,d\pi<\infty,\ 
    C_j(\pi)<\infty,\ j=1,\ldots,m
  \Big\}.
\end{equation}
For $\lambda\in\R_+^m$, recall the aggregate cost $c_\lambda$ defined before Theorem~\ref{thm:finite-budget-otdro}, and set 
  $\varphi_\lambda(x)
  :=
  \sup_{z\in E}\bigl(f(z)-c_\lambda(z,x)\bigr)$.

\begin{lemma}
\label{lem:fixed-marginal-pointwise}
For every $\lambda\in\R_+^m$, the function $\varphi_\lambda$ is universally measurable, its negative part is $P_0$-integrable, and
  $\sup_{\pi\in\Gamma_{\mathrm{fin}}(P_0)}
  \int\bigl(f(z)-c_\lambda(z,x)\bigr)\,\pi(\mathrm dz,\mathrm dx)
  =
  \int\varphi_\lambda(x)\,P_0(\mathrm dx)$.
\end{lemma}

\begin{proof}
The common finite-cost relation $\mathsf C$ defined in section \ref{subsec:otdro} is Borel. For every $a\in\R$, the set
  $\{(z,x)\in\mathsf C:
  f(z)-\sum_{j=1}^m\lambda_jc_j(z,x)>a\}$  is Borel, so its projection onto the second coordinate is analytic. Hence $\varphi_\lambda$ is upper semianalytic and therefore universally measurable. Moreover, $c_\lambda(x,x)=0$ implies 
  $\varphi_\lambda(x)\ge f(x)$,
and consequently $\varphi_\lambda^-\le f^-\in L^1(P_0)$. The upper bound follows pointwise. 

For the reverse inequality, define 
  $\overline c(z,x):=\sum_{j=1}^m c_j(z,x)$,
  $h_\lambda(z,x):=
  f(z)-\sum_{j=1}^m\lambda_jc_j(z,x)$ on $\mathsf C$. For $n\ge1$, let $D_n
  :=
  \{
    (z,x)\in\mathsf C:
    \overline c(z,x)\le n,\ |f(z)|\le n
  \}$, and define
  $\varphi_{\lambda,n}(x)
  :=
  \max\{
    f(x),
    \sup_{z:(z,x)\in D_n}h_\lambda(z,x)
  \}$, with the convention $\sup\varnothing=-\infty$. The sets $D_n$, $n \ge 1$ increase to $\mathsf C$, and therefore 
  $\varphi_{\lambda,n}(x)\uparrow\varphi_\lambda(x)$.
The graph obtained by adjoining the diagonal to $D_n$ is a nonempty Borel relation, and $\varphi_{\lambda,n}(x)\le\max\{f(x),n\}$ is finite. A standard universally measurable $n^{-1}$-selection (cf. the selection argument in Proposition~4 of \cite{blanchetmurthy2019}) gives a map $T_n:E\to E$ such that, outside a $P_0$-null set, either $T_n(x)=x$ or $(T_n(x),x)\in D_n$, and
  $h_\lambda(T_n(x),x)
  \ge
  \varphi_{\lambda,n}(x)-n^{-1}$.

Let $\pi_n=(T_n,\operatorname{id})_\#P_0$. On the set $T_n(x)=x$, all costs vanish; elsewhere, $\overline c(T_n(x),x)\le n$ and $|f(T_n(x))|\le n$. In particular,
  $c_j(T_n(x),x)\le n$, 
  $|f(T_n(x))|\le |f(x)|+n$,
so $\pi_n\in\Gamma_{\mathrm{fin}}(P_0)$. Since $\varphi_{\lambda,n}\ge f$ and $\varphi_{\lambda,n}\uparrow\varphi_\lambda$, monotone convergence applied to $\varphi_{\lambda,n}-f$ yields the reverse inequality and concludes the proof of the lemma: 
\begin{equation*}
\begin{aligned}
  \sup_{\pi\in\Gamma_{\mathrm{fin}}(P_0)}
  \int\bigl(f(z)-c_\lambda(z,x)\bigr)\,\mathrm d\pi
  &\ge
  \liminf_{n\to\infty}
  \int h_\lambda(T_n(x),x)\,P_0(\mathrm dx) 
  \ge
  \lim_{n\to\infty}
  \left(
    \int\varphi_{\lambda,n}\,\mathrm dP_0-\frac1n
  \right)
  =
  \int\varphi_\lambda\,\mathrm dP_0.
\end{aligned}
\end{equation*}
\end{proof}

\begin{proof}[Proof of Theorem~\ref{thm:finite-budget-otdro}]
The path-law formulation is equivalent to the coupling formulation
\begin{equation}
  V(\delta)
  :=
  \sup\left\{
    F(\pi):
    \pi\in\Gamma(P_0),\
    C_j(\pi)\le\delta_j,\ j=1,\ldots,m
  \right\}.
\end{equation}
For every $\lambda\in\R_+^m$, the weak duality gives
  $V(\delta)
  \le
  \lambda\cdot\delta
  +
  \int\varphi_\lambda\,\mathrm dP_0$.
If $V(\delta)=+\infty$, the weak duality implies that every dual value is $+\infty$, and the result follows. We therefore assume that $V(\delta)<\infty$.

For $b\in(0,\infty)^m$, define the budget-value function
  $v(b)
  :=
  \sup\{
    F(\pi):
    \pi\in\Gamma(P_0),\
    C_j(\pi)\le b_j,\ j=1,\ldots,m
  \}$. Let $\pi^0=(\operatorname{id},\operatorname{id})_\#P_0$ be the diagonal coupling and let
  $f_0:=F(\pi^0)=\int f\,\mathrm dP_0$.
The function $v$ is concave, by mixing feasible couplings, and coordinate-wise nondecreasing. It is also finite throughout $(0,\infty)^m$. Indeed, for fixed $b\in(0,\infty)^m$, set
  $\alpha
  :=
  \min\{
    1,\min_{1\le j\le m}{\delta_j}{b_j}^{-1}
  \}>0$.
If $\pi$ is feasible for $v(b)$ and its objective is not $-\infty$, then $\alpha\pi+(1-\alpha)\pi^0$ is feasible for $V(\delta)$. Consequently,
  $\alpha F(\pi)+(1-\alpha)f_0
  \le
  V(\delta)$,
which gives a finite upper bound for $v(b)$; the diagonal coupling gives the finite lower bound $v(b)\ge f_0$. 
Thus, $v$ is a finite concave function on the open convex set $(0,\infty)^m$. By the supporting-hyperplane theorem for finite concave functions (cf. Theorem 23.4 of \cite{rockafellar1970convexanalysis}), it admits a supergradient $\lambda^\star\in\R^m$ at $\delta$:
\begin{equation}
  v(b)
  \le
  v(\delta)+\lambda^\star\cdot(b-\delta),
  \qquad b\in(0,\infty)^m.
\end{equation}
Since $v$ is coordinate-wise nondecreasing, necessarily $\lambda^\star\in\R_+^m$: applying the supergradient inequality at $\delta+t e_j$, for $t>0$, gives $0\le v(\delta+t e_j)-v(\delta)\le t\lambda_j^\star$.

Now fix $\pi\in\Gamma_{\mathrm{fin}}(P_0)$ and let
  $a=(C_1(\pi),\ldots,C_m(\pi))\in\R_+^m$.
For every $\varepsilon>0$, the coupling $\pi$ is feasible for $v(a+\varepsilon\mathbf 1)$. The supergradient inequality therefore gives
  $F(\pi)
  \le
  v(\delta)
  +
  \lambda^\star\cdot
  (a+\varepsilon\mathbf 1-\delta)$.
Letting $\varepsilon\downarrow0$ and taking the supremum over $\Gamma_{\mathrm{fin}}(P_0)$ yields
\begin{equation}
\begin{aligned}
  \lambda^\star\cdot\delta
  +
  \int\varphi_{\lambda^\star}\,\mathrm dP_0
  &=
  \lambda^\star\cdot\delta
  +
  \sup_{\pi\in\Gamma_{\mathrm{fin}}(P_0)}
  \Big(
    F(\pi)-\sum_{j=1}^m\lambda_j^\star C_j(\pi)
  \Big) \le
  v(\delta),
\end{aligned}
\end{equation}
where the equality follows from Lemma~\ref{lem:fixed-marginal-pointwise}. Combined with weak duality, this proves strong duality and shows that $\lambda^\star$ attains the dual infimum.
\end{proof}

\section{A Step-2 Realizability Gap}
\label{app:strictness}

This appendix gives a minimal step-2 example showing that the ambient feature space relaxation in Definition~\ref{def:feature-law-sets} can be strict due to the loss of signature realizability in the relaxed Euclidean feature space. In other words, the inequalities in Proposition \ref{prop:affine-expectation} and Theorem \ref{thm:scalar-reduction} can be strict. Let $\cX$ be the space of one-dimensional bounded-variation paths with $x_0=0$, with the usual time augmentation understood. Select the two state-channel words  $\cW=\{X,XX\}$,
  $u=S^X(\widehat x)$,
  $v=S^{XX}(\widehat x)$,
and write $\Phi(x)=(u,v)\in\R^2$. Recall that the one-dimensional shuffle identity gives $v=u^2/2$ for every realizable selected feature vector. Conversely, every value of \(u\) is realized by the straight-line path with terminal increment \(u\). Let $P_0$ be the Dirac law at the zero path, so $\mu_0=\delta_{(0,0)}$. Since the baseline feature law is a Dirac mass, the common-coupling constraint is trivial in this example. Here, for a coordinate-wise budget vector $(\delta_1,\delta_2)$, the relaxed feature-space ambiguity set $\cV_{(\delta_1,\delta_2)}(\mu_0)$ consists of laws $\nu$ on $\R^2$ satisfying 
  $\int |u|\,\nu(\mathrm du\,\mathrm dv)\le\delta_1$,
  $\int |v|\,\nu(\mathrm du\, \mathrm dv)\le\delta_2$. 
Beyond that, the exact path-level feature laws are precisely the feasible laws in this set that are supported on the parabola $v=u^2/2$.

\begin{proposition}[Strictness of the ambient relaxation]
\label{prop:strictness-step-two}
Assume $ 0<\delta_2<\frac12\delta_1^2$. In the above setup, the ambient feature-space relaxation is strict for the affine payoff \(u\), and it is also strict for the halfspace event \(H_b=\{(u,v):u\ge b\}\) whenever \(\sqrt{2\delta_2}<b\le\delta_1\). More precisely, we have 
\begin{equation}
  \sup_{\nu\in\cR_\Phi\cap\cV_{(\delta_1,\delta_2)}(\mu_0)}
  \int u\,d\nu
  =
  \sqrt{2\delta_2}
  <
  \delta_1
  =
  \sup_{\nu\in\cV_{(\delta_1,\delta_2)}(\mu_0)}
  \int u\,d\nu,
\end{equation}
and
\begin{equation}
  \sup_{\nu\in\cR_\Phi\cap\cV_{(\delta_1,\delta_2)}(\mu_0)}
  \nu(H_b)
  =
  \frac{2\delta_2}{b^2}
  <
  1
  =
  \sup_{\nu\in\cV_{(\delta_1,\delta_2)}(\mu_0)}
  \nu(H_b).
\end{equation}
\end{proposition}

\begin{proof}
Let $\nu\in\cR_\Phi\cap\cV_{(\delta_1,\delta_2)}(\mu_0)$. Since $\nu$ is supported on $v=u^2/2$,
  $\int u^2/2\,\mathrm d\nu
  =
  \int |v|\,\mathrm d\nu
  \le
  \delta_2$.
Hence
  $\int u\, \mathrm d\nu
  \le
  (\int u^2\,\mathrm d\nu)^{1/2}
  \le
  \sqrt{2\delta_2}$. This upper bound is obviously achieved by the straight-line path whose terminal increment is \(\sqrt{2\delta_2}\). In the relaxed feature space, the upper bound
  $\int u\,\mathrm d\nu\le\int |u|\,\mathrm d\nu\le\delta_1$
is attained by the non-realizable feature law \(\nu=\delta_{(\delta_1,0)}\). 

For the half-space event \(H_b=\{(u,v):u\ge b\}\), the exact feasible laws satisfy 
  $\nu(H_b)
  \le
  {\int u^2\,d\nu}b^{-2}
  \le
  {2\delta_2}{b^{-2}}$.
This bound is attained by mixing the zero path with the straight-line path whose terminal increment is \(b\), placing mass \(2\delta_2/b^2\) on the latter. The second-coordinate budget is saturated, and the first-coordinate budget is feasible because
${2\delta_2}/{b}<b\le\delta_1$. In the relaxed feature space, \(\nu=\delta_{(b,0)}\) belongs to \(\cV_{(\delta_1,\delta_2)}(\mu_0)\) and assigns probability one to \(H_b\), since \(b\le\delta_1\). This proves the stated strictness.
\end{proof}

The example isolates the realizability gap. The ambient relaxation is precisely because it forgets the algebraic relations among selected signature coordinates. Nevertheless, one notices that the relaxation gap in the above example can be eliminated if one chooses it precisely $\delta_2=\delta_1^2/2$. However, such eliminations do not automatically carry over to more general cases.

\section{Analytic Coordinate-wise Sig-OT Budget Bounds}
\label{app:analytic-bounds}

This appendix records two preliminary model-based upper bounds for admissible coordinate-wise Sig-OT budget vectors. They are not sharp, but provide further evidence for the coverage of the ambiguity sets via some simple couplings. Throughout this appendix, once a finite selected word set $\cW$ is fixed, the normalizers $r_I>0$, $I\in\cW$, are deterministic and fixed.

We first isolate the simple principle used throughout the appendix. Let $P_0,P_1\in\cP(\cX)$ and let $\pi\in\Pi(P_1,P_0)$ be a coupling of stress and baseline paths. Whenever the selected coordinates have finite first moments under this coupling, define
  $\delta_I^\pi
  :=
  \int |\phi_I(x)-\phi_I(y)|\,\pi(\mathrm dx,\mathrm dy)$,
$I\in\cW $.

\begin{lemma}[Budget certification by a common coupling]
\label{lem:common-coupling-certification}
With $\delta^\pi=(\delta_I^\pi)_{I\in\cW}$ defined above, one has
  $P_1\in\cU_{\delta^\pi}(P_0)$.
Moreover, for every affine direction $\ell=(\ell_I)_{I\in\cW}$,
  $\kappa(\ell,\delta^\pi)
  =
  \int
  \sum_{I\in\cW}|\ell_I|\,|\phi_I(x)-\phi_I(y)|
  \,\pi(\mathrm dx,\mathrm dy)$.
\end{lemma}

\begin{proof}
The first assertion is Definition~\ref{def:sigot-ambiguity}, with the coupling $\pi$ satisfying all coordinate constraints simultaneously. The identity for $\kappa(\ell,\delta^\pi)$ follows by substituting the definition of $\delta_I^\pi$ and using the fact that $\cW$ is finite.
\end{proof}

\subsection{Diffusion perturbations}
\label{subsec:app-diffusion-bounds}

We now instantiate Lemma~\ref{lem:common-coupling-certification} with a synchronous Brownian coupling. Fix $T>0$, a truncation level $N\ge1$, and a finite selected word set 
  $\cW\subset\cA^{\le N}\setminus\{\emptyset\}$. Let $(X_0,Y_0)$ be given on a common filtered probability space, and let $X$ and $Y$ solve the stochastic differential equations 
$\mathrm dX_t=b(X_t)\,\mathrm dt+\sum_{\ell=1}^m\sigma_\ell(X_t)\,\mathrm dW_t^\ell$,
  $\mathrm dY_t=\widetilde b(Y_t)\,\mathrm dt+\sum_{\ell=1}^m\widetilde\sigma_\ell(Y_t)\,\mathrm dW_t^\ell $, $t \ge 0 $, 
driven by the {\it same} $m$-dimensional Brownian motion. Assume that both drift and diffusion coefficients are globally Lipschitz with the common linear-growth constant $M > 0 $. Define 
\begin{equation} \label{eq:epscoef}
  \varepsilon_{\mathrm{coef}}
  :=
  \sup_x\frac{|b(x)-\widetilde b(x)|}{1+|x|}
  +
  \sum_{\ell=1}^m
  \sup_x\frac{|\sigma_\ell(x)-\widetilde\sigma_\ell(x)|}{1+|x|}.
\end{equation}
Set $Q_N:=2\cdot 4^N$ and we assume
  $\|X_0\|_{L^{Q_N}}+\|Y_0\|_{L^{Q_N}}\le K$. All normalizers $r_I>0$ are fixed.

\begin{proposition}[Diffusion perturbations and Stratonovich coordinates]
\label{prop:diffusion-coordinate-bound}
Under the above setup, 
\begin{equation}
  \E|\phi_I(X)-\phi_I(Y)|
  \le
  C_I\Bigl(\|X_0-Y_0\|_{L^{Q_N}}+\varepsilon_{\mathrm{coef}}\Bigr),\quad I\in\cW , 
\end{equation}
where $\varepsilon_{\text{coef}}$ is defined in \eqref{eq:epscoef}, $C_I$ depends on $T,N,d,m,M,K$ and the fixed normalizers $(r_I)_{I\in\cW}$.
\end{proposition}

\begin{proof}
Let $\widehat X_t=(t,X_t)$ and $\widehat Y_t=(t,Y_t)$. We write $C$ for constants that may change from line to line but depend only on the quantities listed in the statement. Set 
  $\Delta_0:=\|X_0-Y_0\|_{L^{Q_N}}+\varepsilon_{\mathrm{coef}}$. For a word $\alpha=(\alpha_1,\ldots,\alpha_k)$ over the time-augmented alphabet, write
\begin{equation}
  J_X^\alpha(t)
  :=
  \int_{0<u_1<\cdots<u_k<t}
  \circ \mathrm d\widehat X^{\alpha_1}_{u_1}\cdots \circ \mathrm d\widehat X^{\alpha_k}_{u_k},
  \qquad
  J_X^\emptyset(t)=1,
\end{equation}
and also define $J_Y^\alpha$ analogously. Then $\phi_I(X)=J_X^I(T)/r_I$.
The proof first identifies the exact It\^o coefficients of these
coordinates and then closes simultaneous moment and stability estimates
over the length of the word.

\noindent \emph{Step 1: Synchronous SDE estimates.}
For every $2\le q\le Q_N$, the standard synchronous-coupling estimates give
\begin{align}
  \|X^*\|_{L^q}+\|Y^*\|_{L^q}&\le C,
  \label{eq:app-sde-moment}, \quad 
  \|(X-Y)^*\|_{L^q}\le C\Delta_0,
\end{align}
where $X^*=\sup_{0\le t\le T}|X_t|$ and similarly for $Y$. Indeed,
  $|b(x)-\widetilde b(y)|
  +
  \sum_{\ell=1}^m|\sigma_\ell(x)-\widetilde\sigma_\ell(y)|
  \le
  C|x-y|+C\varepsilon_{\mathrm{coef}}(1+|y|)$, 
and \eqref{eq:app-sde-moment} 
follow from the standard argument with the Burkholder-Davis-Gundy (BDG) inequality and Gr\"onwall lemma.

\noindent \emph{Step 2: the It\^o structure of the signature coordinates.}
For a nonempty word $\alpha$, let $\alpha^-$ denote the word obtained by deleting its last letter. The Stratonovich recursion is
  $\mathrm dJ_X^\alpha(t)
  =
  J_X^{\alpha^-}(t)\circ \mathrm d\widehat X_t^{\alpha_k}$.
We write
  $\mathrm dJ_X^\alpha(t)
  =
  A_X^\alpha(t)\,\mathrm dt
  +
  \sum_{\ell=1}^m B_{X,\ell}^\alpha(t)\,\mathrm dW_t^\ell$, and set $B_{X,\ell}^{\emptyset}=0$ for every $\ell$, and use the analogous
convention for $Y$.
To state these coefficients explicitly, we define
  $a^{qj}(x):=\sum_{\ell=1}^m\sigma_\ell^q(x)\sigma_\ell^j(x)$,
  $\widetilde a^{qj}(x):=
\sum_{\ell=1}^m\widetilde\sigma_\ell^q(x)\widetilde\sigma_\ell^j(x)$.
If $\alpha_k=0$, then
  $A_X^\alpha=J_X^{\alpha^-}$,
  $B_{X,\ell}^\alpha=0$. Now suppose $\alpha_k=j\in\{1,\ldots,d\}$. Stratonovich-to-It\^o
conversion gives
\begin{equation}
  B_{X,\ell}^\alpha
  =
  J_X^{\alpha^-}\sigma_\ell^j(X),
  \qquad
  A_X^\alpha
  =
  J_X^{\alpha^-}b^j(X)
  +
  \frac12\sum_{\ell=1}^m
  B_{X,\ell}^{\alpha^-}\sigma_\ell^j(X).
\end{equation}
The correction vanishes if $k=1$ or $\alpha_{k-1}=0$. If instead
$k\ge2$ and $\alpha_{k-1}=q\in\{1,\ldots,d\}$,
  $B_{X,\ell}^{\alpha^-}
  =
  J_X^{\alpha^{--}}\sigma_\ell^q(X)$,
  $\sum_{\ell=1}^m
  B_{X,\ell}^{\alpha^-}\sigma_\ell^j(X)
  =
  J_X^{\alpha^{--}}a^{qj}(X)$ with
$\alpha^{--}=(\alpha_1,\ldots,\alpha_{k-2})$. 
Consequently, the drift coefficient is given explicitly by
\begin{equation}
  A_X^\alpha
  =
  \begin{cases}
    J_X^{\alpha^-},
      & \alpha_k=0,\\[2pt]
    J_X^{\alpha^-}b^j(X),
      & \alpha_k=j,\ \text{and } k=1 \text{ or } \alpha_{k-1}=0,\\[2pt]
    J_X^{\alpha^-}b^j(X)
    +\dfrac12J_X^{\alpha^{--}}a^{qj}(X),
      & \alpha_k=j,\ \alpha_{k-1}=q\in\{1,\ldots,d\}.
  \end{cases}
\end{equation}
The same formulas hold for $Y$ with tilded coefficients. In particular,
the only products to be estimated are a lower-order signature coordinate
times $1$, one SDE coefficient, or one covariance coefficient.

\noindent \emph{Step 3: Pointwise coefficient estimates.}
For $0\le k\le N$, set
  $U_k^X:=\sup_{0\le t\le T}\sum_{|\beta|\le k}|J_X^\beta(t)|$,
  $U_k^Y:=\sup_{0\le t\le T}\sum_{|\beta|\le k}|J_Y^\beta(t)|$,
and $V_k:=\sup_{0\le t\le T}
  \sum_{|\beta|\le k}|J_X^\beta(t)-J_Y^\beta(t)|$. In this step, all estimates are pointwise in time, and time arguments are suppressed. Let
  $R:=1+|X|+|Y|$,
  $\Delta:=|X-Y|+\varepsilon_{\mathrm{coef}}$.
The assumptions on the coefficients imply
\begin{equation}
\begin{aligned}
  |b(X)|+\sum_{\ell=1}^m|\sigma_\ell(X)|
  \le C(1+|X|),\quad 
  |b(X)-\widetilde b(Y)|
  +\sum_{\ell=1}^m|\sigma_\ell(X)-\widetilde\sigma_\ell(Y)|
  \le CR\Delta.
\end{aligned}
\end{equation}
Moreover, from the definition of $a^{qj}$,
  $|a^{qj}(X)|\le C(1+|X|)^2$,
  $|a^{qj}(X)-\widetilde a^{qj}(Y)|
  \le CR^2\Delta$.

For any lower-order word $\gamma$ and any pair of coefficient functions
$f,\widetilde f$, the product difference decomposes as
\begin{equation}
  J_X^\gamma f(X)-J_Y^\gamma\widetilde f(Y)
  =
  \bigl(J_X^\gamma-J_Y^\gamma\bigr)f(X)
  +
  J_Y^\gamma\bigl(f(X)-\widetilde f(Y)\bigr).
\end{equation}
For the drift and diffusion terms in a word of length $k$, this identity
with $f=b^j$ or $f=\sigma_\ell^j$ gives
\begin{equation}
\begin{aligned}
  |J_X^{\alpha^-}f(X)|
  \le C(1+|X|)U_{k-1}^X, \quad 
  |J_X^{\alpha^-}f(X)-J_Y^{\alpha^-}\widetilde f(Y)|
  \le
  CR\left[V_{k-1}+U_{k-1}^Y\Delta\right].
\end{aligned}
\end{equation}
When the correction term is present, $k\ge2$, and the same identity with
$f=a^{qj}$ gives
\begin{equation}
\begin{aligned}
  |J_X^{\alpha^{--}}a^{qj}(X)|
  \le C(1+|X|)^2U_{k-2}^X,\quad 
  |J_X^{\alpha^{--}}a^{qj}(X)
    -J_Y^{\alpha^{--}}\widetilde a^{qj}(Y)|
  \le
  CR^2\left[V_{k-2}+U_{k-2}^Y\Delta\right].
\end{aligned}
\end{equation}
For a time-ending word, the corresponding bounds follow directly from
$A_X^\alpha=J_X^{\alpha^-}$ and $B_{X,\ell}^\alpha=0$.
Since $U_{k-2}^X\le U_{k-1}^X$, $V_{k-2}\le V_{k-1}$,
$U_{k-2}^Y\le U_{k-1}^Y$, and $R^2\le R^k$ for $k\ge2$, the explicit
coefficient formulas in Step 2 yield, for every $|\alpha|=k$,
\begin{align}
  |A_X^\alpha|
  +
  \sum_{\ell=1}^m|B_{X,\ell}^\alpha|
  &\le
  C(1+|X|)^k(1+U_{k-1}^X),
  \label{eq:app-coeff-moment}\\
  |A_X^\alpha-A_Y^\alpha|
  +
  \sum_{\ell=1}^m|B_{X,\ell}^\alpha-B_{Y,\ell}^\alpha|
  &\le
  CR^k
  \left[
    V_{k-1}+(1+U_{k-1}^Y)\Delta
  \right].
  \label{eq:app-coeff-stability}
\end{align}

\noindent \emph{Step 4: closing the moment induction.}
Set $p_k:=2\cdot4^{N-k}$, so that $p_0=Q_N$ and $p_N=2$. The factor-four
loss between successive levels leaves enough integrability for the
three-factor product in the stability estimate below. We prove by
induction that, for $0\le k\le N$,
\begin{align}
  \|U_k^X\|_{L^{p_k}}+\|U_k^Y\|_{L^{p_k}}
  &\le C,
  \label{eq:app-Uk}\\
  \|V_k\|_{L^{p_k}}
  &\le C\Delta_0.
  \label{eq:app-Vk}
\end{align}
The case $k=0$ is immediate. Suppose that the bounds hold at level $k-1$.
Applying BDG to the coordinate equations and using
\eqref{eq:app-coeff-moment} gives
  $\|U_k^X\|_{L^{p_k}}
  \le
  C\left\|(1+X^*)^k(1+U_{k-1}^X)\right\|_{L^{p_k}}$,
and similarly for $Y$. The H\"older's inequality, \eqref{eq:app-sde-moment},
and the induction hypothesis yield \eqref{eq:app-Uk}, since
$2kp_k\le Q_N$ and $2p_k\le p_{k-1}$.

For stability, BDG and \eqref{eq:app-coeff-stability} give
  $\|V_k\|_{L^{p_k}}
  \le
  C\|
    R_k[
      V_{k-1}
      +(1+U_{k-1}^Y)\bigl((X-Y)^*+\varepsilon_{\mathrm{coef}}\bigr)
    ]
  \|_{L^{p_k}}$, 
where $R_k=(1+X^*+Y^*)^k$. The term $R_kV_{k-1}$ follows from the 
H\"older inequality with exponents $2,2$, because
$2kp_k\le Q_N$ and $2p_k\le p_{k-1}$. For the term containing
$(X-Y)^*$, the H\"older inequality with exponents $4,4,2$ gives
  $\|R_k(1+U_{k-1}^Y)(X-Y)^*\|_{L^{p_k}}
  \le
  \|R_k\|_{L^{4p_k}}
  \|1+U_{k-1}^Y\|_{L^{4p_k}}
  \|(X-Y)^*\|_{L^{2p_k}}$.
This is bounded by $C\Delta_0$ because
$4kp_k\le Q_N$, $4p_k=p_{k-1}$, and $2p_k\le Q_N$.
Finally, the term containing $\varepsilon_{\mathrm{coef}}$ is bounded by  $\varepsilon_{\mathrm{coef}}
  \|R_k\|_{L^{2p_k}}
  \|1+U_{k-1}^Y\|_{L^{2p_k}}
  \le C\Delta_0$. This proves \eqref{eq:app-Vk}.
At $k=N$, \eqref{eq:app-Vk} and $p_N=2$ imply
  $\E|\phi_I(X)-\phi_I(Y)|
  \le
  r_I^{-1}\|V_N\|_{L^1}
  \le
  C_I\Delta_0$,
$I\in\cW$.
\end{proof}

Let $P_X$ and $P_Y$ denote the path laws of $X$ and $Y$. The synchronous construction is a single coupling of $P_X$ and $P_Y$; hence Lemma~\ref{lem:common-coupling-certification} shows that
  $\delta_I^{\mathrm{sync}}
  :=
  \E|\phi_I(X)-\phi_I(Y)|$,
  $I\in\cW$, is an admissible coordinate-wise Sig-OT budget vector covering $P_Y$ relative to $P_X$.

\subsection{Gaussian fixed-grid calibration}
\label{subsec:app-gaussian-bounds}

We next give a finite-dimensional version adapted to grid-based Gaussian path models. Fix a truncation level $N\ge1$, a finite selected word set
  $\cW\subset\cA^{\le N}\setminus\{\emptyset\}$,
and a grid $0=t_0<\cdots<t_L=T$.
Let 
  $Z=(X_{t_0},\ldots,X_{t_L})\sim N(m,\Sigma)$,
  $\widetilde Z=(Y_{t_0},\ldots,Y_{t_L})\sim N(\widetilde m,\widetilde\Sigma)$
in $\R^{d(L+1)}$. Let $\mathcal I_L(z)$ be the time-augmented piecewise-linear interpolation of the grid vector $z$, and define 
  $\phi_{I,L}(z):={S^I(\mathcal I_L(z))_{0,T}}/{r_I}$, $I\in\cW$,
where $r_I$ are the normalizing constants. For later use, recall the finite-dimensional Gaussian transport formula (cf. \cite{dowsonlandau1982}):
  $W_2^2\bigl(N(m,\Sigma),N(\widetilde m,\widetilde\Sigma)\bigr)
  =
  |m-\widetilde m|^2+
  \operatorname{tr}\!(\Sigma+\widetilde\Sigma
  -2(\widetilde\Sigma^{1/2}\Sigma\widetilde\Sigma^{1/2})^{1/2})$.

\begin{proposition}[Gaussian fixed-grid calibration]
\label{prop:gaussian-grid-bound}
Under the above Gaussian fixed-grid setup, if $(Z,\widetilde Z)$ is an optimal Gaussian $W_2$ coupling, then for every $I\in\cW$,
\begin{equation}
\begin{aligned}
  &\E|\phi_{I,L}(Z)-\phi_{I,L}(\widetilde Z)|  \\
  &\quad\le
  C_{I,L,N}
  \Bigl[
    1+(|m|+\sqrt{\operatorname{tr}\Sigma})^{N-1}
     +( |\widetilde m|+\sqrt{\operatorname{tr}\widetilde\Sigma})^{N-1}
	  \Bigr]
	  W_2\bigl(N(m,\Sigma),N(\widetilde m,\widetilde\Sigma)\bigr)
\end{aligned}
\end{equation}
where $C_{I,L,N}$ also absorbs the fixed normalizer $r_I$.
\end{proposition}

\begin{proof}
By Chen's identity and the tensor-exponential formula for linear segments, every level-$\le N$ signature coordinate of $\mathcal I_L(z)$ is a polynomial in $z$ of degree at most $N$. Hence the mean-value formula gives
\begin{equation}
  |\phi_{I,L}(z)-\phi_{I,L}(z')|
  \le C_{I,L,N}\bigl(1+|z|^{N-1}+|z'|^{N-1}\bigr)|z-z'|.
  \label{eq:poly-lipschitz-coordinate}
\end{equation}
For $N=1$, the polynomial factor is interpreted as a constant.

Let $(Z,\widetilde Z)$ be an optimal coupling for the finite-dimensional Gaussian $W_2$ distance. By \eqref{eq:poly-lipschitz-coordinate} and the Cauchy--Schwarz inequality,
  $\E|\phi_{I,L}(Z)-\phi_{I,L}(\widetilde Z)|
  \le C_{I,L,N}
  \|1+|Z|^{N-1}+|\widetilde Z|^{N-1}\|_{L^2}
  \|Z-\widetilde Z\|_{L^2}$. 
For a Gaussian vector $G\sim N(a,C)$ and each finite $q\ge1$, we use the standard moment estimates:
  $\|G\|_{L^q}\le C_q(|a|+\sqrt{\operatorname{tr}C})$.
Since $\|Z-\widetilde Z\|_{L^2}=W_2(N(m,\Sigma),N(\widetilde m,\widetilde\Sigma))$, the desired bound follows.
\end{proof}

Let $P_L:=(\mathcal I_L)_\#N(m,\Sigma)$ and $\widetilde P_L:=(\mathcal I_L)_\#N(\widetilde m,\widetilde\Sigma)$ be the interpolated path laws. The same optimal Gaussian coupling is used for every selected word; after interpolation, define the simultaneous budget vector:
  $\delta_{I,L}^{\mathrm{G}}
  :=
  \E|\phi_{I,L}(Z)-\phi_{I,L}(\widetilde Z)|$,
  $I\in\cW$. 
It covers $\widetilde P_L$ relative to $P_L$ by Lemma~\ref{lem:common-coupling-certification}.

\section{\texorpdfstring{$J_1$}{J1} Benchmark for Constant-boundary Ruin}
\label{app:j1-ruin-benchmark}

This appendix records the $J_1$ benchmark used in Subsection~\ref{subsubsec:reserve-risk} for the ordinary constant-boundary ruin event. The construction follows Section 3 in \cite{blanchetmurthy2019}. The information is that the robust values for the specific constant-boundary ruin event obtained with this approach, should be expected to be sharp within the class of OT-approaches to the same stress test problem. Let
  $Z_t=\sqrt{\nu m_2}\,B_t-\eta\nu m_1t$,
  $0\le t\le T$,
and let $\mu_Z$ be the law of $Z$ on the space $D([0,T],\R)$. For a capital level $u$, define the set of crossing the constant level $u$: 
  $A_u
  :=
  \{
    z\in D([0,T],\R):
    \sup_{0\le t\le T}z_t\ge u
  \}$.
The Brownian diffusion approximation to the probability of ruin in a finite-horizon is 
  $\psi_B(u,T):=\mu_Z(A_u)$.

The ambiguity set uses the Skorokhod $J_1$ metric. Writing $\Lambda$ for the set of increasing homeomorphisms of $[0,T]$, take 
  $d_{J_1}(z,y)
  :=
  \inf_{\lambda\in\Lambda}
  \{
    \|\lambda-\mathrm{id}\|_\infty
    \vee
    \|z-y\circ\lambda\|_\infty
  \}$,
  $c_{J_1}(z,y)=d_{J_1}(z,y)^p$,
   $p\ge1$.
For the constant-level crossing event, the distance to the set $A_u$ is exactly evaluated by  
  $\inf_{y\in A_u} d_{J_1}(z,y)
  =
  (u-\sup_{0\le t\le T}z_t)_+$,
(cf. Lemma B.2 in \cite{blanchetmurthy2019}) and hence
  $c_{J_1}(z,A_u)
  :=
  \inf_{y\in A_u}c_{J_1}(z,y)
  =
  (u-\sup_{0\le t\le T}z_t)_+^p $. 
Indeed, the upper bound is obtained by shifting the whole path upwards by
\((u-\sup_t z_t)_+\). Conversely, if \(y\in A_u\), then
\(\sup_t y_t\ge u\), and every time change \(\lambda\in\Lambda\) preserves the supremum of \(y\circ\lambda\). Therefore, 
\begin{equation}
  \|z-y\circ\lambda\|_\infty
  \ge
  \sup_t y_t-\sup_t z_t
  \ge
  u-\sup_t z_t,
\end{equation}
which gives the lower bound. Thus,  $J_1$ enlargement of \(A_u\) is a lowering of the ruin threshold. No ambiguity budget is spent on any feature of the path other than closing the vertical shortfall to the constant boundary.

Now we may apply Theorem 2(b) in \cite{blanchetmurthy2019}. For a radius $\delta>0$, define
  $C_u(Z):=c_{J_1}(Z,A_u)
  =
  (u-\sup_{0\le t\le T}Z_t)_+^p$,
and
  $h(r):=\E_{\mu_Z}\![C_u(Z)\,\1_{\{C_u(Z)\le r\}}]$,
  $r\ge0$. If \(\delta\ge\E_{\mu_Z}[C_u(Z)]\), the worst-case probability is \(1\); Otherwise, let
  $r^\star:=h^{-1}(\delta)$ 
with the usual generalized inverse convention. The closed-event OT-DRO formula in Theorem 2(b) in \cite{blanchetmurthy2019} gives
\begin{equation}
\begin{aligned}
  \psi_{\mathrm{rob}}^{J_1}(u,T)
  &:=
  \sup_{Q:\,W_{c_{J_1}}(Q,\mu_Z)\le\delta} Q(A_u) 
  =
  \mu_Z\{C_u(Z)\le r^\star\} 
  =
  \Prob\!\big(
    \sup_{0\le t\le T}Z_t
    \ge u-(r^\star)^{1/p}
  \big)
  =
  \psi_B(\tilde u,T),
\end{aligned}
\end{equation}
where $\tilde u:=u-(r^\star)^{1/p}$. Thus, for the ordinary constant-boundary ruin event, robustification under the \(J_1\) cost is exactly equivalent to lowering the capital level from \(u\) to \(\tilde u\) and evaluating the crossing probability of Brownian motion at that shifted level. This infers precisely that the \(J_1\) benchmark has essentially no additional slack for this task. Among OT-DRO stress tests tailored to the same constant-boundary crossing event, it is the sharp reference case: the distance-to-event calculation depends only on the same scalar functional \(\sup_t z_t\) that defines the ruin. On the other hand, this specific sensitivity suggests that this $J_1$-induced ambiguity set should not be expected to work as a reusable ambiguity set for other general path-dependent events and provide robust bounds that are as informative for them. 

\section{Ambiguity-Ball Coverage and Design Ablations}
\label{app:ablations}
We finally discuss two questions that are deliberately kept distinct in the main text. First, we examine the empirical geometry of the ambiguity budgets: How does the aggregate coordinate-wise radius change under controlled parametric perturbations and under heterogeneous model-family comparisons? Second, we report some options for the numerical choices that enter the downstream robust values.
For a baseline law $P$ and a comparison law $Q$, write the aggregate decoupled empirical radius
  $\Delta_k(P,Q)
  :=
  \|\widehat\delta^{\mathrm{dec}}_k(P,Q)\|_1
  =
  \sum_{I\in\mathcal W_k}\widehat\delta^{\mathrm{dec}}_I(P,Q)$, where $\widehat\delta^{\mathrm{dec}}_I(P,Q)$ is the decoupled empirical $W_1$ budget proxy for signature coordinate $I$ (cf. Section~\ref{subsec:empirical-budget-protocol}). The diagnostics below use the same coordinate normalization and empirical calibration convention as in Section~\ref{sec:numerics}.

For the coverage experiments in Subsections~\ref{subsec:app-parametric-coverage}--\ref{subsec:app-heterogeneous-coverage}, every model is simulated over \(T=1\) using \(10{,}000\) paths, \(2{,}000\) time steps, and five independent seed replicates. We use geometric signatures computed by \texttt{iisignature} \cite{iisignature}, baseline-GBM-fitted MAD scaling with floor \(10^{-8}\), and levels \(k=2,\ldots,6\) for the parametric sweeps and \(k=2,3,4\) for the heterogeneous pairwise experiment. Curves and matrix entries report replicate means; shaded bands in the parametric figures show one cross-replicate standard deviation.

\subsection{Parametric coverage geometry}
\label{subsec:app-parametric-coverage}

The first diagnostic varies one parameter at a time and records the aggregate decoupled radius $\Delta_k$. It illustrates how the Sig-OT ambiguity reacts to parameter uncertainties. Figure~\ref{fig:app-parametric-coverage-raw} reports the raw decoupled radius at the main truncation level $k=4$. The first two panels are within-GBM sanity checks: drift $\mu$ and volatility $\sigma$ are varied around the baseline. The third panel isolates jump-arrival risk through the Merton jump intensity $\lambda$. The fourth and fifth panels are complementary Heston volatility-of-volatility diagnostics: the GBM-to-Heston structural sweep uses the implementation parameter $\xi_v$, while the within-Heston self-sweep follows the older sensitivity notation $\sigma_v$. The final panel isolates local-volatility elasticity through the CEV parameter $\beta$.

The precise sweep specifications are as follows. The GBM baseline is \(\mathrm{GBM}(\mu,\sigma)=\mathrm{GBM}(0.05,0.20)\). The drift sweep fixes \(\sigma=0.20\) and uses \(\mu\in\{-0.10,-0.075,\ldots,0.20\}\); the volatility sweep fixes \(\mu=0.05\) and uses \(\sigma\in\{0.08,0.10,\ldots,0.32\}\). The Merton sweep fixes \((\mu,\sigma,m_J,s_J)=(0.05,0.20,-0.12,0.08)\) and varies \(\lambda\in\{0,0.10,0.20,0.35,0.50,0.75,1,1.25,1.50,1.75,2\}\), so that \(\lambda=0\) recovers the GBM baseline. The GBM--Heston structural sweep compares the GBM baseline with Heston laws having \((\mu,v_0,\theta_v,\kappa_v,\rho)=(0.05,0.04,0.04,3,-0.5)\) and \(\xi_v\in\{0,0.10,0.20,0.30,0.40,0.50,0.65,0.80,1\}\). The Heston self-sweep uses the same parameters with baseline \(\xi_v=0.35\) and varies \(\sigma_v\equiv\xi_v\in\{0,0.10,0.20,0.35,0.50,0.65,0.80,1,1.20\}\). Finally, the CEV dynamics are \(dS_t=0.05S_t\,dt+0.20S_t^\beta\,dW_t\), \(S_0=1\), with baseline \(\beta=1\) and grid \(\beta\in\{0.40,0.50,0.60,0.70,0.80,0.90,1,1.10,1.20,1.30,1.50\}\). Moreover, to compare the shape of the parametric response across truncation levels, Figure~\ref{fig:app-parametric-coverage-normalized} reports the scale-free ratios
\begin{equation}
  \widetilde\Delta_k(P_0,P_\theta)
  :=
  \frac{\Delta_k(P_0,P_\theta)}
  {\max_{\theta'\in\Theta_{\mathrm{sweep}}}\Delta_k(P_0,P_{\theta'})},
  \qquad k=2,\ldots,6.
\end{equation}
This normalization removes the dimension-driven scale differences across $k$. 

\begin{figure}[ht]
\centering
\includegraphics[width=0.99\textwidth]{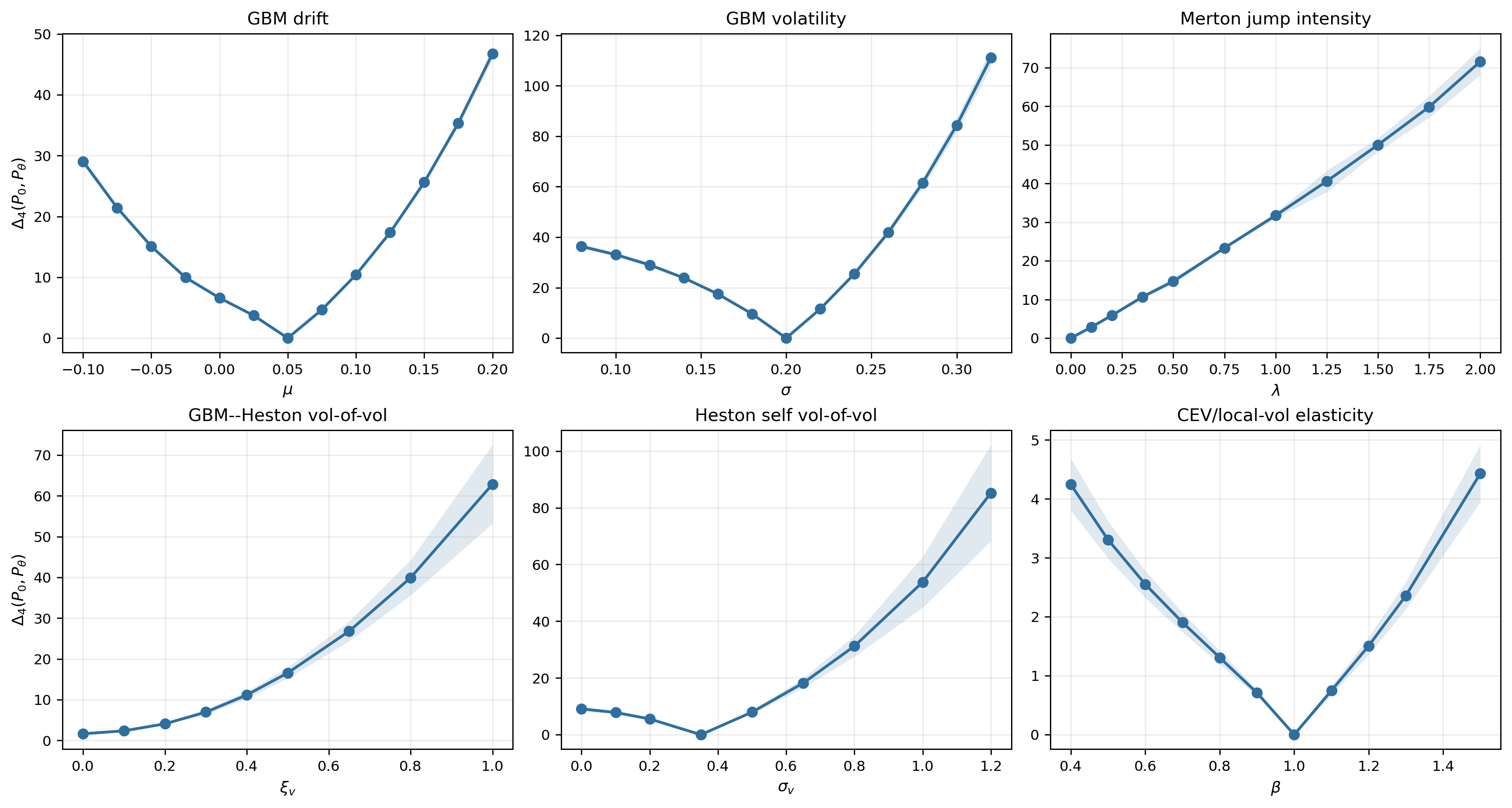}
\caption{Raw parametric coverage scale at $k=4$. Each panel plots the decoupled empirical radius $\Delta_4(P_0,P_\theta)=\|\widehat\delta^{\mathrm{dec}}_4(P_0,P_\theta)\|_1$ against the swept model parameter. The six panels cover GBM drift, GBM volatility, Merton jump intensity, GBM-to-Heston volatility-of-volatility, within-Heston volatility-of-volatility, and CEV/local-volatility elasticity.}
\label{fig:app-parametric-coverage-raw}
\end{figure}

\begin{figure}[ht]
\centering
\includegraphics[width=0.99\textwidth]{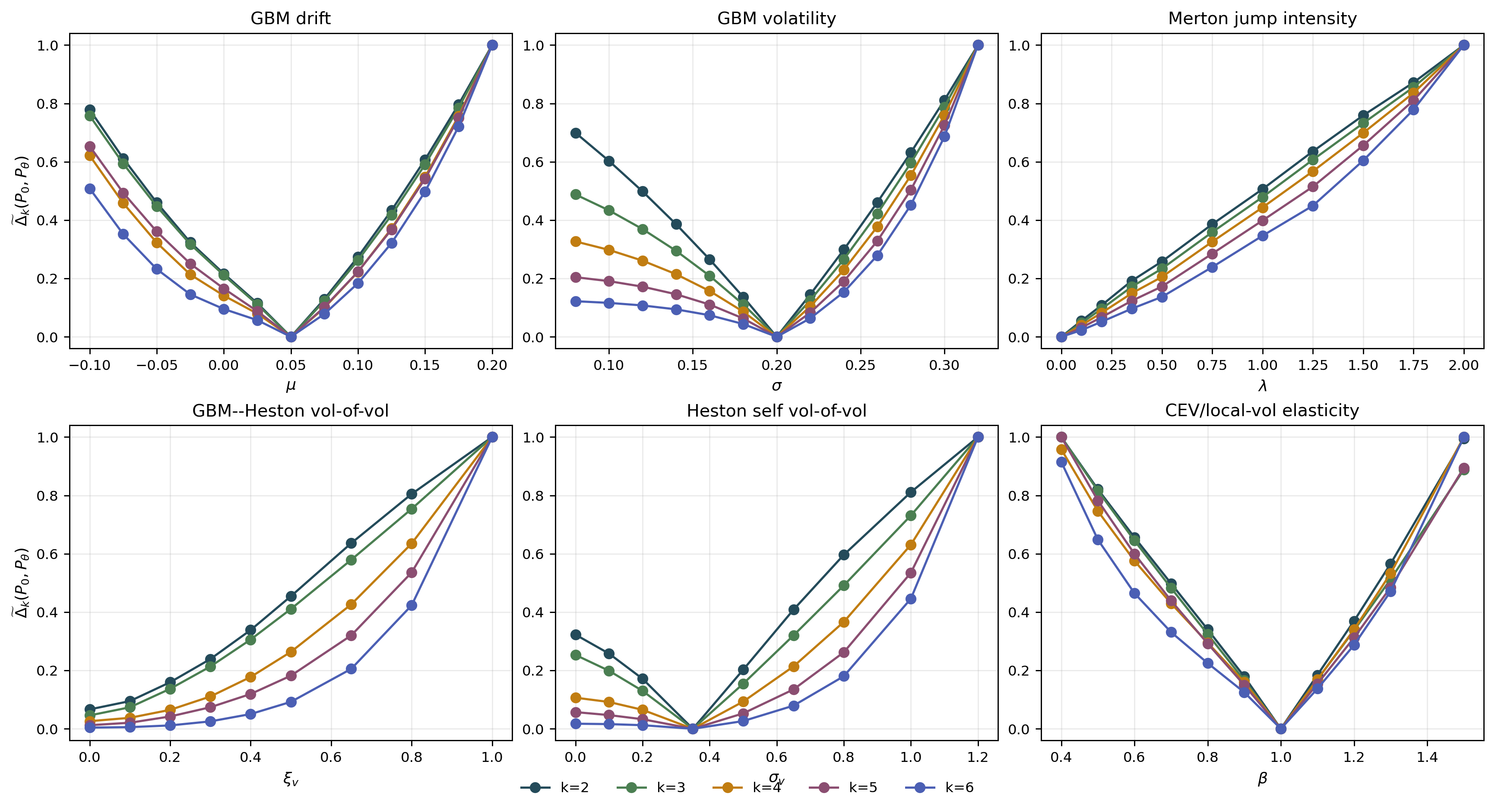}
\caption{Normalized parametric coverage profiles across truncation levels. Each panel plots $\widetilde\Delta_k(P_0,P_\theta)$ for $k=2,\ldots,6$, where each curve is normalized by its maximum value within the same parameter sweep and truncation level. The figure compares the shape of the sensitivity curves.}
\label{fig:app-parametric-coverage-normalized}
\end{figure}

\subsection{Heterogeneous model-family coverage}
\label{subsec:app-heterogeneous-coverage}

The second diagnostic illustrates the heterogeneous coverage of the Sig-OT ambiguity set. For a library of log-price laws 
  $\{P_i\}_{i\in\mathcal I}
  =
  \{\text{GBM},\ \text{GBM-vol},\ \text{Merton},\ \text{CEV},\ \text{Heston},\ \text{Regime}\}$,
we compute the pairwise budget radii
  $D_{ij}^{(k)}
  :=
  \Delta_k(P_i,P_j)$.
Figure~\ref{fig:app-pairwise-coverage} reports the pairwise  matrix at $k=4$ under two feature paths. The first uses the time-augmented log-price path $(t,X_t)$; the second augments it by realized variation, $(t,X_t,\mathrm{RV}_t)$. 
The model library consists of the following six laws with similar variance structures but different model types:
\begin{itemize}[leftmargin=2em,itemsep=0pt,topsep=2pt]
\item GBM with \((\mu,\sigma)=(0.05,0.20)\), and GBM-vol with \((\mu,\sigma)=(0.05,0.22)\);
\item Merton with \((\mu,\sigma,\lambda,m_J,s_J)=(0.05,0.1865,0.35,-0.10,0.07)\);
\item CEV with \(dS_t=0.05S_t\,dt+0.20S_t^{0.60}\,dW_t\) and \(S_0=1\);
\item Heston with \((\mu,v_0,\theta_v,\kappa_v,\xi_v,\rho)=(0.05,0.04,0.04,3,0.35,-0.5)\);
\item Regime mixture \(0.90\,\mathrm{GBM}(0.05,0.20)+0.10\,\mathrm{GBM}(-0.05,0.30)\), with one regime drawn for each whole path.
\end{itemize}
The realized-variation channel is \(\mathrm{RV}_{t_j}=\sum_{r=1}^{j}(X_{t_r}-X_{t_{r-1}})^2\). All pairwise radii use one common scaling convention fitted to the baseline GBM feature bank, so the displayed matrices are symmetric and directly comparable within each feature representation.

\begin{figure}[ht]
\centering
\begin{subfigure}{0.47\textwidth}
  \centering
  \includegraphics[width=\textwidth]{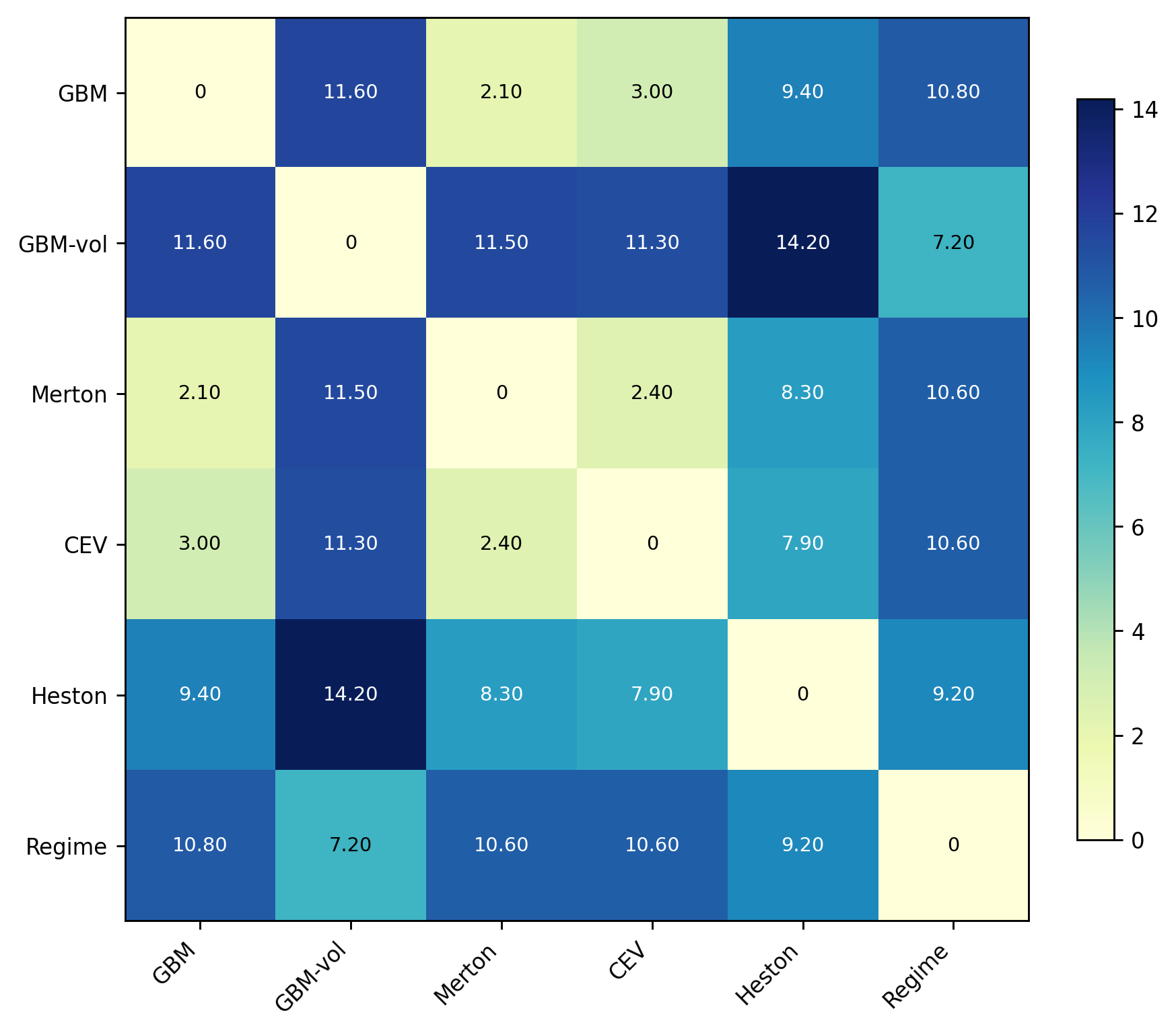}
  \caption{Time--price path}
\end{subfigure}
\begin{subfigure}{0.47\textwidth}
  \centering
  \includegraphics[width=\textwidth]{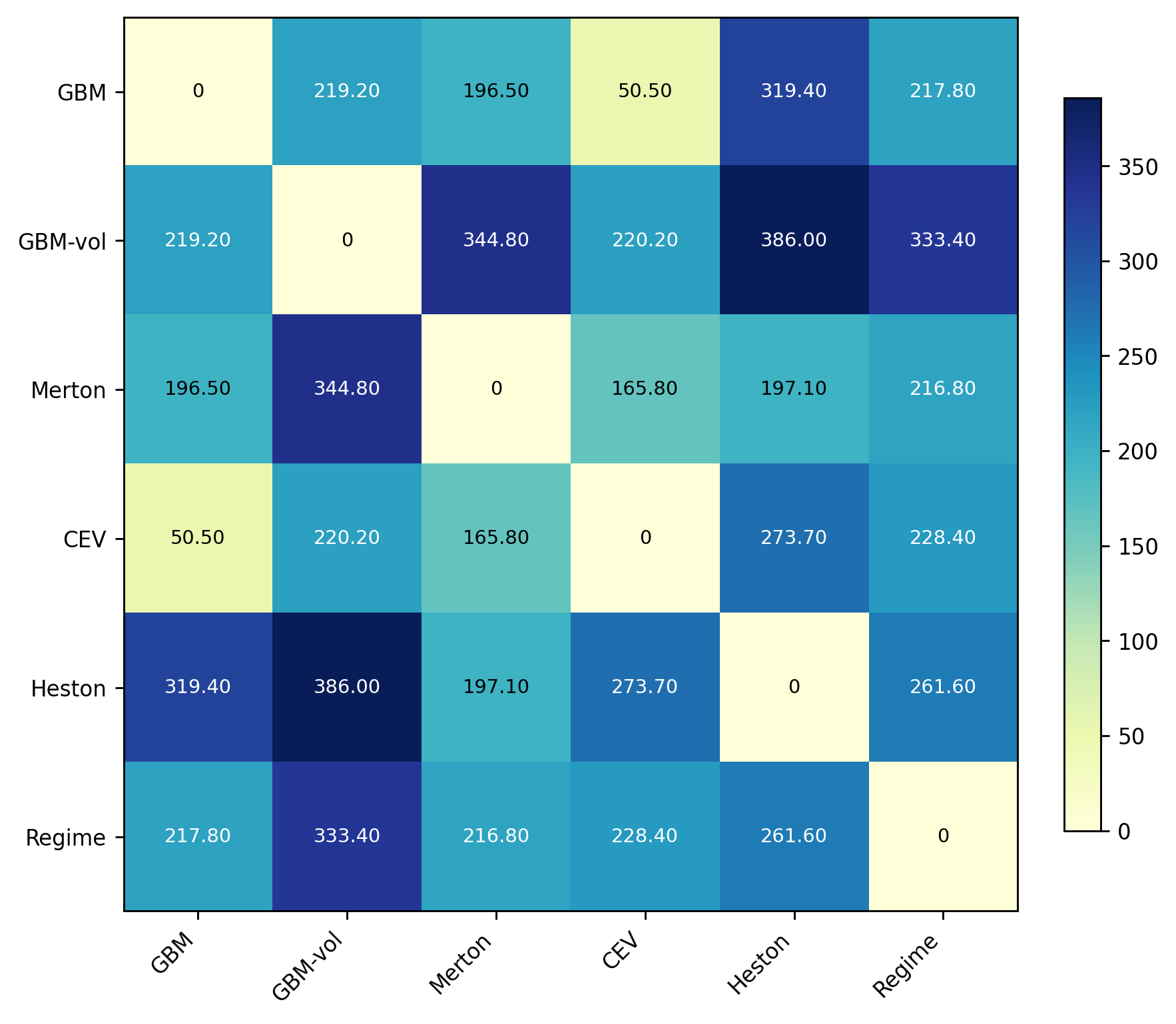}
  \caption{Time--price--realized-variation path}
\end{subfigure}
\caption{Heterogeneous coverage of the Sig-OT ambiguity ball. The entries are pairwise empirical radii $D_{ij}^{(4)}=\|\widehat\delta^{\mathrm{dec}}_4(P_i,P_j)\|_1$ for the model library $\{\mathrm{GBM},\mathrm{GBM\mbox{-}vol},\mathrm{Merton},\mathrm{CEV},\mathrm{Heston},\mathrm{Regime}\}$. The two panels compare the basic time--price representation with a realized-variation-augmented representation.}
\label{fig:app-pairwise-coverage}
\end{figure}

For the basic \((t,X_t)\) representation, the smallest \(k=4\) separations occur among GBM, Merton, and CEV, whereas Heston, the volatility-shifted GBM, and the regime mixture are more distant. Adding realized variation materially enlarges the separation of volatility and stochastic-volatility alternatives: the largest entry becomes the GBM-vol/Heston pair. More generally speaking, the choice of feature path can qualitatively alter the ordering of Sig-OT distances across model pairs, underscoring its role in specifying the Sig-OT path-space geometry.

\subsection{Lookback-option ablations for numerical design choices}
\label{subsec:app-design-ablations}

The final supplementary experiment collects two ablations for the lookback-call pricing experiment from Subsection~\ref{subsubsec:option-pricing}: the budget-weighted lasso parameter $\alpha$ and the signature truncation level $k$.

The reported columns are the same as in the main option-pricing table: baseline price, stress price, stress shift, the learned robust premium $\widehat\kappa$, the tightness ratio $|\mathrm{Shift}|/\widehat\kappa$, and out-of-sample fit quality. In the tables, Ratio denotes $|\mathrm{Shift}|/\widehat\kappa$. For compactness, Tables~\ref{tab:app-lookback-alpha-ablation}--\ref{tab:app-lookback-k-ablation} report the three strikes $K\in\{95,100,105\}$ and the two representative volatility stresses $\sigma\in\{0.15,0.25\}$. Unless explicitly varied, the experiment uses the setup with $\alpha=1$, $k=4$, and the time--price representation $(t,X_t)$.

Each row is the mean of five independent replicates. Each replicate uses \(20{,}000\) paths per law for budget calibration, \(20{,}000\) baseline paths for fitting with \(16{,}000\) training observations, \(500{,}000\) paths per law for price evaluation, and \(2{,}000\) time steps. The task-weighted joint OT calculation uses Sinkhorn regularization \(10^{-2}\), batch size \(512\), and eight subsampled batches per outer iteration.

\begin{table}[p]
\centering
\scriptsize
\setlength{\tabcolsep}{3pt}
\renewcommand{\arraystretch}{0.88}
\caption{Lookback-call ablation for the budget-weighted lasso parameter.}
\label{tab:app-lookback-alpha-ablation}
\begin{tabular}{@{}lrrrrrrrrr@{}}
\toprule
\textcolor{black}{Setting} & \textcolor{black}{$\sigma$} & \textcolor{black}{$\Delta\sigma$} & \textcolor{black}{$K$} & \textcolor{black}{Base} & \textcolor{black}{Stress} & \textcolor{black}{Shift} & \textcolor{black}{$\widehat\kappa$} & \textcolor{black}{Ratio} & \textcolor{black}{Fit $R^2$} \\
\midrule
$\alpha=0$ & 0.15 & -0.05 & 95 & 22.428 & 18.118 & -4.310 & 13.363 & 0.323 & 0.973 \\
$\alpha=0$ & 0.25 & +0.05 & 95 & 22.428 & 26.881 & 4.452 & 15.779 & 0.282 & 0.973 \\
$\alpha=0$ & 0.15 & -0.05 & 100 & 17.527 & 13.217 & -4.310 & 13.363 & 0.323 & 0.973 \\
$\alpha=0$ & 0.25 & +0.05 & 100 & 17.527 & 21.980 & 4.452 & 15.779 & 0.282 & 0.973 \\
$\alpha=0$ & 0.15 & -0.05 & 105 & 13.156 & 8.962 & -4.194 & 12.519 & 0.335 & 0.973 \\
$\alpha=0$ & 0.25 & +0.05 & 105 & 13.156 & 17.535 & 4.379 & 14.679 & 0.298 & 0.973 \\
\midrule
$\alpha=1$ & 0.15 & -0.05 & 95 & 22.428 & 18.118 & -4.310 & 5.794 & 0.745 & 0.963 \\
$\alpha=1$ & 0.25 & +0.05 & 95 & 22.428 & 26.881 & 4.452 & 5.855 & 0.761 & 0.961 \\
$\alpha=1$ & 0.15 & -0.05 & 100 & 17.527 & 13.217 & -4.310 & 5.794 & 0.745 & 0.963 \\
$\alpha=1$ & 0.25 & +0.05 & 100 & 17.527 & 21.980 & 4.452 & 5.855 & 0.761 & 0.961 \\
$\alpha=1$ & 0.15 & -0.05 & 105 & 13.156 & 8.962 & -4.194 & 5.753 & 0.730 & 0.962 \\
$\alpha=1$ & 0.25 & +0.05 & 105 & 13.156 & 17.535 & 4.379 & 5.673 & 0.772 & 0.959 \\
\midrule
$\alpha=2$ & 0.15 & -0.05 & 95 & 22.428 & 18.118 & -4.310 & 5.259 & 0.821 & 0.957 \\
$\alpha=2$ & 0.25 & +0.05 & 95 & 22.428 & 26.881 & 4.452 & 4.786 & 0.931 & 0.951 \\
$\alpha=2$ & 0.15 & -0.05 & 100 & 17.527 & 13.217 & -4.310 & 5.259 & 0.821 & 0.957 \\
$\alpha=2$ & 0.25 & +0.05 & 100 & 17.527 & 21.980 & 4.452 & 4.786 & 0.931 & 0.951 \\
$\alpha=2$ & 0.15 & -0.05 & 105 & 13.156 & 8.962 & -4.194 & 5.105 & 0.823 & 0.954 \\
$\alpha=2$ & 0.25 & +0.05 & 105 & 13.156 & 17.535 & 4.379 & 4.565 & 0.960 & 0.948 \\
\midrule
$\alpha=3$ & 0.15 & -0.05 & 95 & 22.428 & 18.118 & -4.310 & 4.670 & 0.925 & 0.947 \\
$\alpha=3$ & 0.25 & +0.05 & 95 & 22.428 & 26.881 & 4.452 & 4.001 & 1.113 & 0.942 \\
$\alpha=3$ & 0.15 & -0.05 & 100 & 17.527 & 13.217 & -4.310 & 4.670 & 0.925 & 0.947 \\
$\alpha=3$ & 0.25 & +0.05 & 100 & 17.527 & 21.980 & 4.452 & 4.001 & 1.113 & 0.942 \\
$\alpha=3$ & 0.15 & -0.05 & 105 & 13.156 & 8.962 & -4.194 & 4.458 & 0.943 & 0.943 \\
$\alpha=3$ & 0.25 & +0.05 & 105 & 13.156 & 17.535 & 4.379 & 3.670 & 1.194 & 0.935 \\
\bottomrule
\end{tabular}
\end{table}

\begin{table}[p]
\centering
\scriptsize
\setlength{\tabcolsep}{3pt}
\renewcommand{\arraystretch}{0.82}
\caption{Lookback-call ablation for signature truncation depth.}
\label{tab:app-lookback-k-ablation}
\begin{tabular}{@{}lrrrrrrrrr@{}}
\toprule
\textcolor{black}{Setting} & \textcolor{black}{$\sigma$} & \textcolor{black}{$\Delta\sigma$} & \textcolor{black}{$K$} & \textcolor{black}{Base} & \textcolor{black}{Stress} & \textcolor{black}{Shift} & \textcolor{black}{$\widehat\kappa$} & \textcolor{black}{Ratio} & \textcolor{black}{Fit $R^2$} \\
\midrule
$k=2$ & 0.15 & -0.05 & 95 & 22.428 & 18.118 & -4.310 & 4.485 & 0.963 & 0.924 \\
$k=2$ & 0.25 & +0.05 & 95 & 22.428 & 26.881 & 4.452 & 4.779 & 0.936 & 0.924 \\
$k=2$ & 0.15 & -0.05 & 100 & 17.527 & 13.217 & -4.310 & 4.485 & 0.963 & 0.924 \\
$k=2$ & 0.25 & +0.05 & 100 & 17.527 & 21.980 & 4.452 & 4.779 & 0.936 & 0.924 \\
$k=2$ & 0.15 & -0.05 & 105 & 13.156 & 8.962 & -4.194 & 4.408 & 0.954 & 0.923 \\
$k=2$ & 0.25 & +0.05 & 105 & 13.156 & 17.535 & 4.379 & 4.680 & 0.940 & 0.923 \\
\midrule
$k=3$ & 0.15 & -0.05 & 95 & 22.428 & 18.118 & -4.310 & 5.442 & 0.795 & 0.951 \\
$k=3$ & 0.25 & +0.05 & 95 & 22.428 & 26.881 & 4.452 & 5.630 & 0.793 & 0.951 \\
$k=3$ & 0.15 & -0.05 & 100 & 17.527 & 13.217 & -4.310 & 5.442 & 0.795 & 0.951 \\
$k=3$ & 0.25 & +0.05 & 100 & 17.527 & 21.980 & 4.452 & 5.630 & 0.793 & 0.951 \\
$k=3$ & 0.15 & -0.05 & 105 & 13.156 & 8.962 & -4.194 & 5.450 & 0.772 & 0.952 \\
$k=3$ & 0.25 & +0.05 & 105 & 13.156 & 17.535 & 4.379 & 5.586 & 0.786 & 0.951 \\
\midrule
$k=4$ & 0.15 & -0.05 & 95 & 22.428 & 18.118 & -4.310 & 5.794 & 0.745 & 0.963 \\
$k=4$ & 0.25 & +0.05 & 95 & 22.428 & 26.881 & 4.452 & 5.855 & 0.761 & 0.961 \\
$k=4$ & 0.15 & -0.05 & 100 & 17.527 & 13.217 & -4.310 & 5.794 & 0.745 & 0.963 \\
$k=4$ & 0.25 & +0.05 & 100 & 17.527 & 21.980 & 4.452 & 5.855 & 0.761 & 0.961 \\
$k=4$ & 0.15 & -0.05 & 105 & 13.156 & 8.962 & -4.194 & 5.753 & 0.730 & 0.962 \\
$k=4$ & 0.25 & +0.05 & 105 & 13.156 & 17.535 & 4.379 & 5.673 & 0.772 & 0.959 \\
\midrule
$k=5$ & 0.15 & -0.05 & 95 & 22.428 & 18.118 & -4.310 & 5.682 & 0.764 & 0.965 \\
$k=5$ & 0.25 & +0.05 & 95 & 22.428 & 26.881 & 4.452 & 5.501 & 0.810 & 0.961 \\
$k=5$ & 0.15 & -0.05 & 100 & 17.527 & 13.217 & -4.310 & 5.682 & 0.764 & 0.965 \\
$k=5$ & 0.25 & +0.05 & 100 & 17.527 & 21.980 & 4.452 & 5.501 & 0.810 & 0.961 \\
$k=5$ & 0.15 & -0.05 & 105 & 13.156 & 8.962 & -4.194 & 5.596 & 0.754 & 0.964 \\
$k=5$ & 0.25 & +0.05 & 105 & 13.156 & 17.535 & 4.379 & 5.383 & 0.814 & 0.960 \\
\midrule
$k=6$ & 0.15 & -0.05 & 95 & 22.428 & 18.118 & -4.310 & 5.958 & 0.725 & 0.967 \\
$k=6$ & 0.25 & +0.05 & 95 & 22.428 & 26.881 & 4.452 & 5.306 & 0.841 & 0.961 \\
$k=6$ & 0.15 & -0.05 & 100 & 17.527 & 13.217 & -4.310 & 5.958 & 0.725 & 0.967 \\
$k=6$ & 0.25 & +0.05 & 100 & 17.527 & 21.980 & 4.452 & 5.306 & 0.841 & 0.961 \\
$k=6$ & 0.15 & -0.05 & 105 & 13.156 & 8.962 & -4.194 & 5.867 & 0.716 & 0.966 \\
$k=6$ & 0.25 & +0.05 & 105 & 13.156 & 17.535 & 4.379 & 5.223 & 0.840 & 0.959 \\
\bottomrule
\end{tabular}
\end{table}

Table~\ref{tab:app-lookback-alpha-ablation} displays the expected approximation--robustness tradeoff. With no budget penalty (\(\alpha=0\)), the fit is strongest (\(R^2\simeq0.973\)). Increasing \(\alpha\) shrinks the fitting accuracy but increases the mean tightness ratio from \(0.307\) at \(\alpha=0\) to \(0.752\), \(0.881\), and \(1.036\) at \(\alpha=1,2,3\), respectively, while the mean \(R^2\) decreases from \(0.973\) to \(0.962\), \(0.953\), and \(0.942\). At \(\alpha=3\), the three upward-volatility rows have ratios \(1.113\), \(1.113\), and \(1.194\). These rows show the warning that an aggressively budget-penalized surrogate need not cover the shift of the original nonlinear payoff if $\alpha$ is too large.

Table~\ref{tab:app-lookback-k-ablation} shows that the shallow representation \(k=2\) is the tightest numerically, with a mean ratio \(0.949\), but has a materially weaker approximation quality, with a mean \(R^2=0.924\). Moving to \(k=3\) increases the mean \(R^2\) to \(0.951\), and \(k=4\) increases it further to \(0.962\), while all six \(k=4\) ratios remain between \(0.730\) and \(0.772\). Levels \(k=5\) and \(k=6\) offer only marginal additional fit (\(R^2\simeq0.963\)) and do not improve the mean tightness ratio relative to \(k=4\). Thus, \(k=4\) captures most of the available approximation gain without the additional feature complexity of deeper truncations.

\bibliographystyle{apacite}
\bibliography{main}

\end{document}